\documentclass[runningheads]{llncs}

\usepackage{etex} 

\RequirePackage{rotating}

\usepackage{bm}
\usepackage{amsmath,amsfonts,amssymb}
\usepackage{extarrows}
\usepackage{stmaryrd}
\usepackage{times}
\usepackage[font=footnotesize,labelfont=bf]{caption}
\usepackage[inline]{enumitem}
\usepackage[ruled,vlined,linesnumbered]{algorithm2e}
\SetAlFnt{\scriptsize} 
\usepackage[usenames, table, svgnames, dvipsnames]{xcolor}
\usepackage{url}
\usepackage{pdfpages} 
\usepackage{bigfoot}
\DeclareNewFootnote[plain]{default}
\usepackage[title,page]{appendix}

\usepackage{booktabs}
\usepackage{array}
\usepackage{multirow}
\newcolumntype{P}[1]{>{\raggedleft\arraybackslash}p{#1}}

\usepackage{pgf}
\usepackage{tikz}
\usetikzlibrary{arrows,shapes,snakes,automata,backgrounds,petri,positioning}

\usepackage{float}

\usepackage{graphicx}

\usepackage{marvosym}

\usepackage[normalem]{ulem}  
\newcommand\redsout{\bgroup\markoverwith{\textcolor{red}{\rule[0.5ex]{2pt}{0.4pt}}}\ULon}  

\def\ackname{Data Availability Statement}%
\def\acknowledgement{\par\addvspace{17pt}\small\rmfamily
	\trivlist\if!\ackname!\item[]\else
	\item[\hskip\labelsep
	{\bfseries\ackname}]\fi}

\newcommand{\oomit}[1]{}

\usepackage{todonotes}

\makeatletter
\RequirePackage[bookmarks,unicode,colorlinks=true]{hyperref}%
\def\@citecolor{blue}%
\def\@urlcolor{blue}%
\def\@linkcolor{blue}%

\makeatother

\begin{document}

\title{Learning One-Clock Timed Automata\thanks{This work has been partially funded by NSFC under grant No.~61625206, 61972284, 61732001 and 61872341, by the ERC Advanced Project FRAPPANT under grant No.~787914, and by the CAS Pioneer Hundred Talents Program under grant No.~Y9RC585036.}}

\titlerunning{Learning One-Clock Timed Automata}

\author{
Jie An\inst{1}$^{\text{(\Letter)}}$
\and 
Mingshuai Chen\inst{2,3,4}
\and 
Bohua Zhan\inst{3,4}
\and \\
Naijun Zhan\inst{3,4}$^{\text{(\Letter)}}$
\and
Miaomiao Zhang\inst{1}$^{\text{(\Letter)}}$
}
\authorrunning{J.\ An et al.}

\institute{
School of Software Engineering, Tongji University, Shanghai, China\\
\email{\{1510796,miaomiao\}@tongji.edu.cn}
\and
Lehrstuhl f{\"u}r Informatik 2, RWTH Aachen University, Aachen, Germany\\
\email{chenms@cs.rwth-aachen.de}
\and
State Key Lab. of Computer Science, Institute of Software, CAS, Beijing, China  \\
\email{\{bzhan,znj\}@ios.ac.cn}
\and 
University of Chinese Academy of Sciences, Beijing, China
}

\maketitle


\setcounter{footnote}{0}

\setlength{\floatsep}{1\baselineskip}
\setlength{\textfloatsep}{1\baselineskip}
\setlength{\intextsep}{1\baselineskip}

\begin{abstract}
  We present an algorithm for active learning of deterministic timed automata with a single clock. The algorithm is within the framework of Angluin's $L^*$ algorithm and inspired by existing work on the active learning of symbolic automata. Due to the need of guessing for each transition whether it resets the clock, the algorithm is of exponential complexity in the size of the learned automata. Before presenting this algorithm, we propose a simpler version where the teacher is assumed to be \emph{smart} in the sense of being able to provide the reset information. We show that this simpler setting yields a polynomial complexity of the learning process. Both of the algorithms are implemented and evaluated on a collection of randomly generated examples. We furthermore demonstrate the simpler algorithm on the functional specification of the TCP protocol.

\end{abstract}

\section{Introduction}\label{sec:introduction}
In her seminal work~\cite{Angluin87}, Angluin introduced the $L^*$ algorithm for learning a regular language from queries and counterexamples within a query-answering framework. The Angluin-style learning therefore is also termed \emph{active learning} or \emph{query learning}, which is distinguished from \emph{passive learning}, i.e., generating a model from a given data set. Following this line of research, an increasing number of efficient active learning methods (cf.~\cite{Vaandrager17}) have been proposed to learn, e.g., Mealy machines~\cite{ShahbazG09,MargariaNRS04}, I/O automata~\cite{AartsV10}, register automata~\cite{HowarSJC12,AartsFKV15,CasselHJS16}, nondeterministic finite automata~\cite{BolligHKL09}, B\"uchi automata~\cite{Farzan08,LiCZL17}, symbolic automata~\cite{Maler14,Drews17,ArgyrosD18} and Markov decision processes~\cite{TapplerA0EL19}, to name just a few. Full-fledged libraries, tools and applications are also available for automata-learning tasks~\cite{Bollig10,Isberner15,Fiterau16,Fiterau17}.

For real-time systems where timing constraints play a key role, however, learning a formal model is much more complicated. As a classical model for real-time systems, timed automata~\cite{Alur94} have an infinite set of timed actions. 
This yields a fundamental difference to finite automata featuring finite alphabets. Moreover, it is difficult to detect resets of clock variables from observable behaviors of the system. This makes learning formal models of timed systems a challenging yet interesting problem.

Various attempts have been carried out in the literature on learning timed models, which can be classified into two tracks. The first track pursues active learning methods, e.g.~\cite{Grinchtein10} for learning event-recording automata (ERA)~\cite{Alur99} and \cite{scis2020} for learning real-time automata (RTA)~\cite{Dima01}. ERA are time automata where, for every untimed action $a$, a clock is used to record the time of the last occurrence of $a$. The underlying learning algorithm~\cite{Grinchtein10}, however, is prohibitively complex due to too many degrees of freedom and multiple clocks for recording events. RTA are a class of special timed automata with one clock to record the execution time of each action by resetting at the starting. 
The other track pursues passive learning. In \cite{Verwer07,VerwerWW12}, an algorithm was proposed to learn deterministic RTA. The basic idea is that the learner organizes a tree sketching traces of the data set while merging nodes of the tree following a certain heuristic function. A passive learning algorithm for timed automata with one clock was further proposed in~\cite{VerwerWW09,VerwerWW11}. A common weakness of passive learning methods is that the generated model merely accepts all positive traces while it rejects all negative ones for the given set of traces, without guaranteeing that it is a correct model of the target system. A theoretical result was established in \cite{VerwerWW11} showing it is possible to obtain the target system by continuously enriching the data set, however the number of iterations is unknown. In addition, the passive learning methods cited above concern only discrete-time semantics of the underlying timed models, i.e., the clock takes values from non-negative integers. We furthermore refer the readers to~\cite{CaldwellCF16,PastoreMM17} for learning specialized forms of practical timed systems in a passive manner,~\cite{TapplerALL19} for passively learning timed automata using genetic programming which scales to automata of large sizes, \cite{Schmidt13} for learning probabilistic real-time automata incorporating clustering techniques in machine learning, and~\cite{TapplerA0EL19} for $L^*$-based learning of Markov decision processes with testing and sampling.

In this paper, we present the first active learning method for deterministic one-clock timed automata (DOTAs) under continuous-time semantics\footnote{The proposed learning method applies trivially to discrete-time semantics too.}. Such timed automata provide simple models while preserving adequate expressiveness, and therefore have been widely used in practical real-time systems~\cite{StiggeEGY11,AbdullahDM018,DenningS81}. We present our approach in two steps. First, we describe a simpler algorithm, under the assumption that the teacher is \emph{smart} in the sense of being able to provide information about clock resets in membership and equivalence queries. The basic idea is as follows. We define the \emph{reset-logical-timed language} of a DOTA and show that the timed languages of two DOTAs are equivalent if their reset-logical-timed languages are equivalent, which reduces the learning problem to that of learning a reset-logical-timed language.Then we show how to learn the reset-logical-timed language following Maler and D'Antoni's learning algorithms for symbolic automata~\cite{Maler14,Drews17}. We claim the correctness, termination and polynomial complexity of this learning algorithm. Next, we extend this algorithm to the case of a normal teacher. The main difference is that the learner now needs to \emph{guess} the reset information on transitions discovered in the observation table. Due to these guesses, the latter algorithm features exponential complexity in the size of the learned automata. The proposed learning methods are implemented and evaluated on randomly generated examples. We also demonstrate the simpler, polynomial algorithm on a practical case study concerning the functional specification of the TCP protocol. Detailed proofs for theorems and lemmas in this paper can be found in Appendix~\ref{appendix_proof}.

In what follows, Sect.~\ref{sec:preliminaries} provides
preliminary definitions on one-clock timed automata. The learning
algorithm with a smart teacher is presented and analyzed in
Sect.~\ref{sec:smartteacher}. We then present the situation with a
normal teacher in Sect.~\ref{sec:normalteacher}. The experimental
results are reported in Sect.~\ref{sec:experiments}. Finally,
Sect.~\ref{sec:conclusion} concludes this paper.

\section{Preliminaries}\label{sec:preliminaries}
Let $\mathbb{R}_{\geq 0}$ and $\mathbb{N}$ be the set of non-negative
reals and natural numbers, respectively, and $\mathbb{B}$ the Boolean
set. We use $\top$ to stand for $\text{true}$ and $\bot$ for
$\text{false}$. The projection of an $n$-tuple $\mathbf{x}$ onto its
first two components is denoted by $\Pi_{\{1, 2\}}\mathbf{x}$, which
extends to a sequence of tuples as
$\Pi_{\{1,2\}}(\mathbf{x}_1,\ldots,\mathbf{x}_k) =
\left(\Pi_{\{1,2\}}\mathbf{x}_1, \ldots,
  \Pi_{\{1,2\}}\mathbf{x}_k\right)$.

\oomit{
\subsection{Learning deterministic finite automata}\label{sbsc:Lstar}
We start by briefly reviewing the $L^{*}$ algorithm for learning regular sets from membership queries and equivalence queries. Angluin proved in~\cite{Angluin87} that the class of regular languages could be learned efficiently, namely with time polynomial in the size of the canonical deterministic finite automaton for the target language.

\oomit{
\begin{definition}[Deterministic Finite Automata]
	A deterministic finite automaton (DFA) is a 5-tuple $\mathcal{A} = (Q, \Sigma, \delta, q_0, F)$ where $Q$ is a non-empty finite set of states; $\Sigma$ is a finite alphabet; $\delta: Q \times \Sigma \rightarrow Q$ is the transition relation, a partial function on $Q \times \Sigma$; $q_0 \in Q$ is the initial state and $F \subseteq Q$ is the set of accepting (final) states.
\end{definition}

A \emph{word} over $\Sigma$ is a finite sequence $\omega = \sigma_1 \sigma_2 \dots \sigma_n$, where $\sigma_i \in \Sigma$ for $i = 1,2,\dots,n$. $|\omega|=n$ is the length of $\omega$. $\epsilon$ is the empty word with length $|\epsilon|=0$. A word $\omega$ is called an \emph{action} if $|\omega|=0$ or $|\omega|=1$. $\Sigma^{*}$ is the set of words over $\Sigma$. The transition function $\delta$ can be extended to $\hat{\delta}:Q\times \Sigma^{*} \rightarrow Q$, where $\hat{\delta}(q,\epsilon)=q$, and $\hat{\delta}(q,\omega\cdot \sigma)=\hat{\delta}(\hat{\delta}(q,\omega),\sigma)$ for $q\in Q$, $\sigma\in\Sigma$, and $\omega\in\Sigma^{*}$. A word $\omega \in \Sigma^{*}$ is \emph{accepted} by $\mathcal{A}$ if $\hat{\delta}(q_0,\omega)\in F$. Without causing ambiguity, we also denote $\delta(q,\sigma)=q'$ as $(q,\sigma,q')\in\delta$.
For a transition $(q, \sigma, q') \in \delta$, $q$ and $q'$ are called the \emph{source state} and \emph{target state} of the transition respectively. 
}
In the $L^{*}$ algorithm, a \emph{learner} is designed to construct a deterministic finite automaton (DFA) which recognizes the unknown target language $\mathcal{L}$ by asking a teacher questions. The \emph{teacher} knows the target language $\mathcal{L}$ and can answer the learner's questions, categorized into two types of queries including the \emph{membership query}, i.e., ``Is the word $\omega$ in $\mathcal{L}$ ?", and the \emph{equivalence query}, i.e., ``Is the recognized language $\mathcal{L'}$ of the current hypothesis equal to $\mathcal{L}$ ?". The whole procedure is sketched as follows. The learner first asks membership queries to gather enough information about the automata, which is stored in an observation table. When the table is closed and consistent, the learner constructs a DFA as a hypothesis of the target language. Then he asks an equivalence query. If the teacher's answer is positive, the learner will be sure that the hypothesis indeed recognizes the target language $\mathcal{L}$ and the algorithm terminates. Otherwise, the learner receives a counterexample $ctx$ from the teacher, and will make membership queries guided by the counterexample to gather more information to construct a new hypothesis. Repeat the above procedure until an automaton recognizing the target language is learned. 
}
\oomit{
\begin{definition}[Observation Table]
	An observation table for a DFA $\mathcal{A}$ is a 6-tuple $T = (\Sigma, S, R, E, f, row)$ where $\Sigma$ is a finite alphabet; $S,R,E \subset \Sigma^{*}$ are finite sets, $S$ is called the set of prefixes, $R$ is called the boundary, and $E$ is called the set of suffixes; $s\cdot\sigma \in R$ for $\forall s\in S, \sigma \in \Sigma$; $S, R$ are disjoint\footnote{$\uplus$ : disjoint union of two sets}: $S\cup R = S\uplus R$; $S\cup R$ is a prefix-closed set; $f : (S\cup R)\cdot E \rightarrow \{-,+\}$ is a classification function such that for a word $\omega \cdot e \in (S\cup R)\cdot E$, $f(\omega\cdot e) = -$ if $\omega \cdot e \notin \mathcal{L}(\mathcal{A})$, and $f(\omega\cdot e) = +$ if $\omega \cdot e \in \mathcal{L}(\mathcal{A})$; $row$ is a function to return the vector of $f(\omega\cdot e)$ indexed by $e\in E$ for $\omega\in S\cup R$.
\end{definition}

Before suggesting a hypothesis, the learner asks membership queries to make the observation table $T$ closed and consistent:
\begin{itemize}
	\item closed if for every $r \in R$, there exists $s \in S$ such that $row(s) = row(r)$.
	\item consistent if for every $\omega_1, \omega_2 \in S$, $row(\omega_1)=row(\omega_2)$ implies $row(\omega_1\cdot \sigma) = row(\omega_2\cdot \sigma)$ for $\forall \sigma \in \Sigma$. 
\end{itemize}

If the table is not closed, there is some $r\in R$ such that $row(r)$ is different from $row(s)$ for all $s\in S$. The learner moves the $r$ from $R$ to $S$, adds all words $r\cdot \sigma$ for $\sigma \in \Sigma$ to $R$, and makes membership queries to fill the extended observation table.

If the table is not consistent, one inconsistency is resolved through finding two words $\omega_1, \omega_2 \in S, \sigma \in \Sigma$ and $e\in E$ such that $row(\omega_1) = row(\omega_2)$ and $f(\omega_1 \sigma \cdot e) \neq f(\omega_2 \sigma \cdot e)$ and adding this new suffix $\sigma\cdot e$ to $E$. The learner also needs to fill the extended observation table by making membership queries. The observation table is consistent when no more such words can be found.

If the observation table $T = (\Sigma, S, R, E, f, row)$ is closed and consistent, the learner can construct a hypothesis DFA $H_\mathcal{A} = (Q, \Sigma, \delta, q_0, F)$ where $Q = \{row(s)|s\in S\}$, $F = \{row(s)|f(s\cdot \epsilon) = +\}$, $q_0 = row(\epsilon)$ and $\delta(row(s),\sigma) = row(s \cdot \sigma)$. When receiving a counterexample $ctx$, the learner adds all prefixes of $ctx$ to $S$ and the possible inconsistency should be fixed.
}

\emph{Timed automata} \cite{Alur94}, a kind of finite automata extended with a finite set of real-valued clocks, are widely used to model real-time systems. In this paper, we consider a subclass of timed automata with a single clock, termed \emph{one-clock timed automata} (OTAs). Let $c$ be the clock variable, denote by $\Phi_c$ the set of clock constraints of the form $\phi::= \top \mid c\bowtie m \mid \phi \wedge \phi$, 
where $m\in\mathbb{N}$ and $\bowtie\ \in \{=,<,>,\le,\ge\}$.
\begin{definition}[One-clock timed automata]
A one-clock timed automaton $\mathcal{A}=(\Sigma,Q,\\q_0,F,c,\Delta)$, where $\Sigma$ is a finite set of actions, called the \emph{alphabet}; $Q$ is a finite set of locations; $q_0\in Q$ is the initial location; $F\subseteq Q$ is a set of accepting locations; $c$ is the unique clock; and $\Delta\subseteq Q\times\Sigma\times\Phi_c\times\mathbb{B} \times Q$ is a finite set of transitions.
\oomit{
\begin{itemize}
	\item $\Sigma$ is a finite set of actions, called the \emph{alphabet};
	\item $Q$ is a finite set of locations;
	\item $q_0\in Q$ is the initial location;
	\item $F\subseteq Q$ is a set of accepting locations;
	\item $c$ is the unique clock;
	\item $\Delta\subseteq Q\times\Sigma\times\Phi_c\times\mathbb{B} \times Q$ is a finite set of transitions.
\end{itemize}
}
\end{definition}
%

A transition $\delta = (q,\sigma,\phi,b,q')$ allows a jump from  the \emph{source location} $q$ to the \emph{target location} $q'$ by performing the action $\sigma\in\Sigma$ if the constraint $\phi\in \Phi_c$ is satisfied. Meanwhile, clock $c$ is reset to zero if $b = \top$, and remains unchanged otherwise. 

A \emph{clock valuation} is a function $\nu\colon c \mapsto \mathbb{R}_{\ge 0}$ that assigns a non-negative real number to the clock. For $t\in\mathbb{R}_{\ge 0}$, let $\nu+t$ be the clock valuation with $(\nu+t)(c)=\nu(c)+t$. According to the definitions of clock valuation and clock constraint, a transition \emph{guard} can be represented as an interval whose endpoints are in $\mathbb{N}\cup\{\infty\}$. For example, $\phi_1\colon c<5 \wedge c\ge 3$ is represented as $[3,5)$, $\phi_2\colon c=6$ as $[6,6]$, and $\phi_3\colon \top$ as $[0,\infty)$. We will use the inequality- and interval-representation interchangeably in this paper.

A \emph{state} $s$ of $\mathcal{A}$ is a pair $(q,\nu)$, where $q\in Q$ and $\nu$ is a clock valuation.
A \emph{run} $\rho$ of $\mathcal{A}$ is a finite sequence $\rho=(q_0,\nu_0)\xrightarrow{t_1,\sigma_1}(q_1,\nu_1)\xrightarrow{t_2,\sigma_2}\cdots\xrightarrow{t_n,\sigma_n}(q_n,\nu_n)$, where $\nu_0(c)=0$, $t_i\in\mathbb{R}_{\ge 0}$ stands for the time delay spending on $q_{i-1}$ before  $\delta_i=(q_{i-1},\sigma_i,\phi_i,b_i,q_i)\in\Delta$ is taken, only if  (1) $\nu_{i-1} + t_{i}$ satisfies $\phi_i$, (2) $\nu_{i}(c)=\nu_{i-1}(c)+t_{i}$ if $b_{i}=\bot$, otherwise $\nu_{i}(c)=0$, for all $1\leq i \leq n$.
 A run $\rho$ is \emph{accepting} if $q_n\in F$. 

\oomit{A \emph{timed word} $\omega$ over $\Sigma\times\mathbb{R}_{\ge 0}$ is a finite sequence $\omega=(\sigma_1,\tau_1)(\sigma_2,\tau_2)\cdots(\sigma_n,\tau_n)$ where $\sigma_i\in\Sigma$ and $\tau_i\in\mathbb{R}_{\ge 0}$ for $1 \le i \le n$. $\vert \omega \vert=n$ is the length of $\omega$. We abbreviate $(\epsilon,\tau)$ to $\epsilon$ as the empty timed word for all $\tau\in\mathbb{R}_{\ge 0}$ and let $\vert \epsilon \vert=0$. A timed word $\omega$ is called a \emph{timed action} if $\vert \omega \vert=0$ or $\vert \omega \vert=1$. }

The \emph{\textit{trace}} of a run $\rho$ is a timed word, denoted by $\textit{trace}(\rho)$. $\textit{trace}(\rho)=\epsilon$ if $\rho=(q_0,\nu_0)$, and $\textit{trace}(\rho)=(\sigma_1,t_1)(\sigma_2,t_2)\cdots(\sigma_n,t_n)$ if $\rho=(q_0,\nu_0)\xrightarrow{t_1,\sigma_1}(q_1,\nu_1)\xrightarrow{t_2,\sigma_2}\cdots\xrightarrow{t_n,\sigma_n}(q_n,\nu_n)$. Since $t_i$ is the time delay on $q_{i-1}$, for $1 \leq  i \leq n$, such a timed word is also called \emph{delay-timed word}. The corresponding \emph{reset-delay-timed word} can be defined as $\textit{trace}_r(\rho)=(\sigma_1,t_1,b_1)(\sigma_2,t_2,b_2)\cdots(\sigma_n,t_n,b_n)$, where $b_i$ is the reset indicator for $\delta_{i}$, for $1\leq i \leq n$. If $\rho$ is an accepting run of $\mathcal{A}$, $\textit{trace}(\rho)$ is called an \emph{accepting timed word}. The \emph{recognized timed language} of $\mathcal{A}$ is the set of accepting delay-timed words, i.e.,  $\mathcal{L}(\mathcal{A})=\{\textit{trace}(\rho) \, \vert \, \rho \text{ is an accepting run of } \mathcal{A}\}$. The \emph{recognized reset-timed language} $\mathcal{L}_{r}(\mathcal{A})$ is defined as $\{\textit{trace}_r(\rho) \, \vert \, \rho \text{ is an accepting run of } \mathcal{A}\}$.

The delay-timed word $\omega = (\sigma_1,t_1)(\sigma_2,t_2)\cdots(\sigma_n,t_n)$ is observed outside, from the view of the global clock. On the other hand, the behavior can also be observed inside, from the view of the local clock. This results in a \emph{logical-timed word} of the form $\gamma=(\sigma_1,\mu_1)(\sigma_2,\mu_2)\cdots(\sigma_n,\mu_n)$ with 
$\mu_i = t_i$ if $i=1 \vee b_{i-1} = \top$ and $\mu_i = \mu_{i-1} + t_{i}$ otherwise.
\oomit{
 \begin{equation*}
 \mu_i = \begin{cases}
 t_i, & \text{if}\ i=1\ \text{or}\ b_{i-1} = \top\\
 \mu_{i-1} + t_{i}, & \text{otherwise}.
 \end{cases}
 \end{equation*} 
}
 We will denote the mapping from  delay-timed words to  logical-timed words above by $\Gamma$. 
  
  Similarly, we introduce \emph{reset-logical-timed word} $\gamma_r=(\sigma_1,\mu_1,b_1)(\sigma_2,\mu_2,b_2)$ $\cdots$ $(\sigma_n,\mu_n,b_n)$ 
    as the counterpart of $\omega_r=(\sigma_1,t_1,b_1)(\sigma_2,t_2,b_2)\cdots(\sigma_n,t_n,b_n)$ in terms of the local clock. 
  Without any substantial change, we can  extend the mapping $\Gamma$ to map reset-delay-timed words  to reset-logical-timed words. 
 The \emph{recognized logical-timed language} of  $\mathcal{A}$ is given  as $L(\mathcal{A})=\{\Gamma(\textit{trace}(\rho) )\, \vert \, \rho$ is an accepting run of  $\mathcal{A}\}$, and the \emph{recognized reset-logical-timed language} of  $\mathcal{A}$  as $L_{r}(\mathcal{A})=\{\Gamma(\textit{trace}_r(\rho)) \, \vert \, \rho$  is an accepting run of $\mathcal{A}\}$.

 An OTA is a \emph{deterministic one-clock timed automaton} (DOTA) if
 there is at most one run for a given delay-timed word. In other
 words, for any location $q \in Q$ and action $\sigma\in\Sigma$, the guards
 of transitions outgoing from $q$ labelled with $\sigma$ are disjoint
 subsets of $\mathbb{R}_{\geq 0}$. We say a DOTA is \emph{complete} if
 for any of its location $q \in Q$ and action $\sigma\in\Sigma$, the corresponding guards form a
 partition of $\mathbb{R}_{\geq 0}$.  This means any given delay-timed
 word has exactly one run. Any DOTA $\mathcal{A}$ can be transformed
 into a complete DOTA (referred to as COTA) $\mathbb{A}$ accepting the same timed
 language as follows: (1)\ Augment $Q$ with a ``sink'' location $q_s$ which is not an
 accepting location; (2)\ For every $q \in Q$ and $\sigma \in \Sigma$, if there is no outgoing transition from $q$ labelled with $\sigma$, introduce a (resetting) transition from $q$ to $q_s$ with label $\sigma$ and guard $[0,\infty)$; (3)\ Otherwise, let $S$ be the subset of $\mathbb{R}_{\ge 0}$ not covered by the guards of transitions from $q$ with label $\sigma$. Write $S$ as a union of intervals $I_1,\dots,I_k$ in a minimal way, then introduce a (resetting) transition from $q$ to $q_s$ with label $\sigma$ and guard $I_j$ for each $1\le j\le k$. 
 
 From now on, we therefore assume that we are working with COTAs.
\begin{example}
  Fig.~\ref{fig:dota_cota} depicts the transformation of a DOTA
  $\mathcal{A}$ (left part) into a COTA $\mathbb{A}$ (right
  part). First, a non-accepting ``sink'' location $q_s$ is
  introduced. Second, we introduce three fresh transitions (marked in blue) from $q_1$ to $q_s$ as well as
  transitions from $q_s$ to itself. At last, for location $q_0$ and label $a$, the existing guards
  cover $(1,3)$, with complement $[0,1]\cup [3,\infty)$. Hence, we
  introduce transitions $(q_0,a,[0,1],\top,q_s)$ and
  $(q_0,a,[3,\infty),\top,q_s)$. Two fresh transitions from $q_1$ to $q_s$ are introduced similarly.
\end{example}

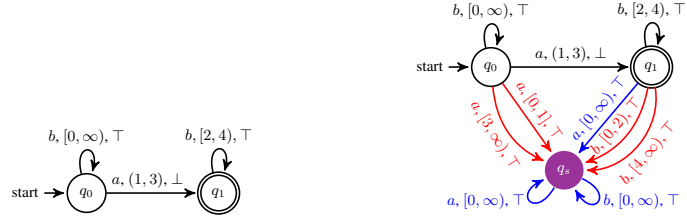
\begin{figure}[!t]
	\centering
	\begin{minipage}{0.4\textwidth}
		\centering
		\vspace*{1.6cm}
		\begin{tikzpicture}[->, >=stealth', shorten >=1pt, auto, node distance=2.7cm, semithick, scale = 0.7,every node/.style={scale=0.65}]
		\node[initial, state]  (0) {$q_0$};
		\node[accepting, state](1) [right of = 0] {$q_1$};
		
		\path (0) edge node[above, pos=.45] {$a$, $(1,3)$, $\bot$} (1) 
		(0) edge [loop above] node[above] {$b$, $[0,\infty)$, $\top$} (0)
		(1) edge [loop above] node[above] {$b, [2, 4), \top$} (1);
		\end{tikzpicture}
		\label{fig:dota}
	\end{minipage}
	\begin{minipage}{0.5\textwidth}
		\centering
		\begin{tikzpicture}[->, >=stealth', shorten >=1pt, auto, node distance=4cm, semithick, scale=0.7, every node/.style={scale=0.65}]
		\node[initial,state] (0) {$q_0$};
		\node[accepting, state] (1) [right = 1.6cm of 0] {$q_1$};
		\node[state, fill = violet!80, draw=none, text=white] (2) [below right= 1cm and .6cm of 0] {$q_s$};
		
		\path  (0) edge node[above] {$a$, $(1,3)$, $\bot$} (1)
		(0) edge [loop above] node[above] {$b$, $[0,\infty)$, $\top$} (0)
		(1) edge [loop above] node[above] {$b$, $[2, 4)$, $\top$} (1)
		(0) edge [red] node[above,sloped]  {$a$, $[0,1]$, $\top$} (2)
		(0) edge [bend right, red] node[below,sloped,pos=.45]  {$a$, $[3,\infty)$, $\top$} (2)
		(1) edge [blue] node[above,sloped,pos=.48]  {$a$, $[0,\infty)$, $\top$} (2)
		(1) edge [bend left, red] node[above,sloped]  {$b$, $[0,2)$, $\top$} (2)
		(1) edge [in=0, out=-80, red] node[below,sloped]  {$b$, $[4,\infty)$, $\top$} (2)
		(2) edge [in= 205, out=-120, loop, blue] node[left, blue] {$a$, $[0,\infty)$, $\top$} (2)
		(2) edge [in= -60, out=-25, loop, blue] node[right] {$b$, $[0,\infty)$, $\top$} (2);
		\end{tikzpicture}
		\label{fig:cota}
	\end{minipage}
	\caption{A DOTA $\mathcal{A}$ on the left and the corresponding COTA $\mathbb{A}$ on the right. The initial location is indicated by `start' and an accepting location is doubly circled.}
	\label{fig:dota_cota}
\end{figure}


\section{Learning from a Smart Teacher}\label{sec:smartteacher}
In this section, we consider the case of learning a COTA $\mathbb{A}$
with a smart teacher. Our learning algorithm relies on the following
reduction of the equivalence over timed languages to that of
reset-logical timed languages.

\begin{theorem}\label{theorem:L=>mathcalL}
  Given two DOTAs $\mathcal{A}$ and $\mathcal{B}$, if $L_r(\mathcal{A})=L_r(\mathcal{B})$, then $\mathcal{L}(\mathcal{A})=\mathcal{L}(\mathcal{B})$.
\end{theorem}

Theorem~\ref{theorem:L=>mathcalL} assures that $L_r(\mathcal{H})=L_r(\mathbb{A})$ implies $\mathcal{L}(\mathcal{H})=\mathcal{L}(\mathbb{A})$, that is, to construct a COTA $\mathbb{A}$ that recognizes a target timed language $\mathcal{L}=\mathcal{L}(\mathbb{A})$, it suffices to learn a \emph{hypothesis} $\mathcal{H}$ which recognizes the same reset-logical timed language. For equivalence queries, instead of checking directly whether $L_r(\mathcal{H}) = L_r(\mathbb{A})$, the contraposition of Theorem~\ref{theorem:L=>mathcalL} guarantees that we can perform equivalence queries over their timed counterparts: if $\mathcal{L}(\mathcal{H}) = \mathcal{L}(\mathbb{A})$, then $\mathcal{H}$ recognizes the target language already; otherwise, a counterexample making $\mathcal{L}(\mathcal{H}) \neq \mathcal{L}(\mathbb{A})$ yields an evidence also for $L_r(\mathcal{H}) \neq L_r(\mathbb{A})$.

We now describe the behavior of the teacher who keeps an
automaton $\mathbb{A}$ to be learnt, while providing knowledge about the automaton
by answering membership and equivalence
queries through an oracle she maintains. For the membership query, the teacher receives a logical-timed
word $\gamma$ and returns whether $\gamma$ is in $L(\mathbb{A})$. In
addition, she is smart enough to return the reset-logical-timed word $\gamma_r$ that corresponds to $\gamma$ (the exact correspondence is described in
Sect. \ref{sbsc:membership}). For the equivalence query, the teacher
is given a hypothesis $\mathcal{H}$ and decides whether
$\mathcal{L}(\mathcal{H}) = \mathcal{L}(\mathbb{A})$. If not, she is smart enough to
return a reset-delayed-timed word $\omega_r$ as a counterexample. The
usual case where a teacher can deal with only standard delay-timed
words will be discussed in Sect.~\ref{sec:normalteacher}.


\begin{remark}
  The assumption that the teacher can respond with timed words coupled
  with reset information is reasonable, in the sense that the learner
  can always infer and detect the resets of the logical clock by
  referring to a global clock on the wall, as long as he can observe running states of $\mathbb{A}$, i.e., observing the clock valuation of the system whenever an event happens therein. This conforms with the idea of combining automata learning with white-box techniques, as exploited in~\cite{HowarJV19}, providing that in many application scenarios source code is available for the analysis.
\end{remark}

\oomit{\begin{definition}[logical timed word, reset logical timed word ]\label{df:valuation_word}
	Given a run $\rho=(q_0,\nu_0)$ $\xrightarrow{t_1,\delta_1}(q_1,\nu_1)$ $\xrightarrow{t_2,\delta_2}\cdots\xrightarrow{t_n,\delta_n}(q_n,\nu_n)$, 
	\begin{itemize}
		\item a \emph{logical timed word} over $\Sigma\times\mathbb{R}_{\geq 0}$ is of the form  $\gamma=(\sigma_1,\mu_1)(\sigma_2,\mu_2)\cdots(\sigma_n,\mu_n)$;
		\item a \emph{reset logical timed word } over $\Sigma\times\mathbb{R}_{\geq 0}\times\mathbb{B}$ is of the form  $\gamma_r=(\sigma_1,\mu_1,b_1)(\sigma_2,\mu_2,b_2)$ $\cdots$ $(\sigma_n,\mu_n,b_n)$;
	\end{itemize}
 where $\sigma_i\in\delta_i$, $\mu_i=\nu_{i-1}(c)+t_i$ and $b_i \in \mathbb{B}$,  for $1\leq i \leq n$ 
\end{definition}

Hence, given a run $\rho$, $\mu_i$ is clock valuation when the action $\sigma_i$ was conducted. $\bm{\sigma}=(\sigma,\mu)$ is called \emph{local timed action}. We call $\bm{\sigma_r}=(\sigma',\mu',b')$ as \emph{reset-local timed action} and denote the \emph{corresponding local timed action} of $\bm{\sigma_r}$  as $\hat{\bm{\sigma_r}}=(\sigma',\mu')$. The following lemma shows how to get an unique reset logical timed word  according to a given reset-delay-timed word.

\begin{lemma}\label{lemma:unique_valuation_word}
	Given a run $\rho$ of a DOTA $\mathcal{A}$, there are an unique reset-delay-timed word $\omega^{d}_{r}=(\sigma_1,t_1, b_1)(\sigma_2,t_2,b_2)\cdots(\sigma_n,t_n,b_n)$, and an unique reset logical timed word  $\gamma_r = (\sigma_1,\mu_1,b_1)(\sigma_2,\mu_2,b_2)\cdots(\sigma_n,\mu_n,b_n)$, where $\mu_1 = t_1$, for $2\leq i \leq n$, 
	\begin{equation*}
		\mu_i = \begin{cases}
		t_i, & \text{if } b_{i-1}=0; \\
		\mu_{i-1} + t_{i}, & \text{if } b_{i-1}\neq 0.
		\end{cases}
	\end{equation*}
\end{lemma}
\begin{proof}
	The fact in this lemma follows immediately from definitions of clock valuation, run, reset-delay-timed word, and reset logical timed word .
\end{proof}

A \emph{mixed-logical timed word} $\gamma_{m}=(\sigma_1,\mu_1,b_1)\cdots(\sigma_k,\mu_k,b_k)(\sigma_{k+1},\mu_{k+1})\cdots(\sigma_n,\mu_n)$ where $0\leq k \leq n$, can be divided into two part that the prefix has reset information and the suffix has no reset information. Hence, the reset logical timed word  and logical timed word can be regarded as two special cases of mixed-logical timed word. A mixed-logical timed word $\gamma_m$ is a \emph{valid} logical timed word of a DOTA $\mathcal{A}$ iff $\gamma_m$ belongs to a run $\rho$ of $\mathcal{A}$. Specially, for a COTA, an invalid mixed-logical timed word can be checked easily if there exists $b_i=\bot$ and $\mu_i > \mu_{i+1}$ for $1\leq i\leq n$. 

\begin{lemma}\label{lemma:unique_run}
Given a valid mixed-logical timed word $\gamma_m$ of a DOTA $\mathcal{A}$, there is an unique reset logical timed word  $\gamma_r$ and it belongs to an unique run $\rho$ of $\mathcal{A}$. 
\end{lemma}
\begin{proof}
	Since $\mathcal{A}$ is deterministic, the fact in this lemma follows immediately.
\end{proof}

Hence, we call $L(\mathcal{A})=\{\gamma_{r}(\rho) \, \vert \, \rho \text{ is an accepting run of } \mathcal{A}\}$ as \emph{recognized reset-valuation language}. 

\begin{theorem}\label{theorem:L<=>mathcalLr}
	Given two DOTAs $\mathcal{A}$ and $\mathcal{B}$, $\mathcal{L}_{r}(\mathcal{A})=\mathcal{L}_{r}(\mathcal{B})$ iff $L(\mathcal{A})=L(\mathcal{B})$.
\end{theorem}

\begin{theorem}\label{theorem:L=>mathcalL}
	Given two DOTAs $\mathcal{A}$ and $\mathcal{B}$, if $L(\mathcal{A})=L(\mathcal{B})$, then $\mathcal{L}(\mathcal{A})=\mathcal{L}(\mathcal{B})$.
\end{theorem}
}

%

In what follows, we elaborate the learning procedure including membership queries, hypotheses construction, equivalence queries and counterexample processing.

\subsection{Membership query}\label{sbsc:membership}

In our setting, the oracle maintained by the smart teacher can be regarded as a COTA $\mathbb{A}$ that recognizes the target timed language $\mathcal{L}$, and thereby its logical-timed language $L(\mathbb{A})$ and reset-logical-timed counterpart $L_r(\mathbb{A})$. In order to collect enough information for constructing a hypothesis, the learner makes membership queries as ``Is the logical-timed word $\gamma$ in $L(\mathbb{A})$?''. If there does not exist a run $\rho$ such that $\Gamma(\textit{trace}(\rho)) = \gamma$, meaning that there is some $k$ such that the run is blocked after the $k$'th action (i.e. $\gamma$ is \emph{invalid}) and hence the teacher gives a negative answer, associated with a reset-logical-timed word $\gamma_r$ where all $b_i$'s with $i > k$ are set to $\top$;
If there exists a run $\rho$ (which is unique due to the determinacy of $\mathbb{A}$) that admits $\gamma$ (i.e., $\gamma$ is \emph{valid}), the teacher answers ``Yes'', if $\rho$ is accepting, or ``No'' otherwise, while in both cases providing the corresponding reset-logical-timed word $\gamma_r$, with $\Pi_{\{1,2\}} \gamma_r = \gamma$.

For the sake of simplicity, we define a function $\pi$ that maps a logical-timed word to its unique reset-logical-timed counterpart in membership queries. Information gathered from the membership queries is stored in a timed observation table defined as follows.

\begin{definition}[Timed observation table]
	A timed observation table for a COTA $\mathbb{A}$ is a 7-tuple $\mathbf{T} = (\Sigma, \bm{\Sigma}, \bm{\Sigma_r}, \bm{S}, \bm{R}, \bm{E}, f)$ where $\Sigma$ is the finite alphabet; $\bm{\Sigma} = \Sigma \times \mathbb{R}_{\geq 0}$ is the infinite set of logical-timed actions; $\bm{\Sigma_r} = \Sigma\times\mathbb{R}_{\geq 0}\times\mathbb{B}$ is the infinite set of reset-logical-timed actions; $\bm{S},\bm{R} \subset \bm{\Sigma_r}^{*}$ and $\bm{E}\subset\bm{\Sigma}^{*}$ are finite sets of words, where $\bm{S}$ is called the set of prefixes, $\bm{R}$ the boundary, and $\bm{E}$ the set of suffixes. Specifically,
	\begin{itemize}
		\item $\bm{S}$ and $\bm{R}$ are disjoint, i.e., $\bm{S}\cup\bm{R} = \bm{S}\uplus\bm{R}$;
		\item The empty word is by default both a prefix and a suffix, i.e., $\epsilon \in \bm{E}$ and $\epsilon \in \bm{S}$;
		\item $f\colon (\bm{S}\cup\bm{R})\cdot \bm{E} \mapsto \{-,+\}$ is a classification function such that for a reset-logical-timed word $\gamma_r$, $\gamma_r \cdot e \in (\bm{S}\cup \bm{R})\cdot \bm{E}$, $f(\gamma_r \cdot e) = -$ if $\Pi_{\{1,2\}}\gamma_r \cdot e$ is invalid, otherwise if $\Pi_{\{1,2\}}\gamma_r \cdot e \notin L(\mathbb{A})$, $f(\gamma_r \cdot e) = -$, and $f(\gamma_r \cdot e) = +$ if $\Pi_{\{1,2\}}\gamma_r \cdot e \in L(\mathbb{A})$.
	\end{itemize}
\end{definition}

Given a table $\mathbf{T}$, we define
$\mathit{row}\colon \bm{S}\cup \bm{R} \mapsto (\bm{E} \mapsto \{+,-\}) $ as a function mapping each $\gamma_r\in \bm{S}\cup \bm{R}$ to a vector
indexed by $e \in \bm{E}$, each of whose components is defined as
$f(\gamma_r \cdot e)$, denoting a potential
location. 

Before constructing a hypothesis $\mathcal{H}$ based on the timed observation table $\mathbf{T}$, the learner has to ensure that $\mathbf{T}$ satisfies the following conditions:
\begin{itemize}
	\item \emph{Reduced:} $\forall s, s'\in \bm{S}\colon s \neq s' \text{ implies } \mathit{row}(s) \neq \mathit{row}(s')$;
	\item \emph{Closed:} $\forall r \in \bm{R}, \exists s \in \bm{S}\colon \mathit{row}(s) = \mathit{row}(r)$;
	\item \emph{Consistent:} $\forall \gamma_r, {\gamma_r}' \in \bm{S}\cup\bm{R}$, $\mathit{row}({\gamma_r})=\mathit{row}({\gamma_r}')$ implies $\mathit{row}({\gamma_r} \cdot {\bm{\sigma_r}}) = \mathit{row}({\gamma_r}'\cdot {\bm{\sigma_r}}')$, for all ${\bm{\sigma_r}},{\bm{\sigma_r}}'\in \bm{\Sigma_r}$ satisfying ${\gamma_r} \cdot {\bm{\sigma_r}}, {\gamma_r}' \cdot {\bm{\sigma_r}}' \in \bm{S}\cup\bm{R}$ and $\Pi_{\{1,2\}}{\bm{\sigma_r}} = \Pi_{\{1,2\}}{\bm{\sigma_r}}'$;
	\item \emph{Evidence-closed:} $\forall s \in \bm{S}$ and $\forall e \in \bm{E}$, the reset-logical-timed word $\pi(\Pi_{\{1,2\}}s \cdot e)$ belongs to $\bm{S}\cup\bm{R}$;
	\item \emph{Prefix-closed:} $\bm{S}\cup\bm{R}$ is prefix-closed.
\end{itemize}

A timed observation table $\mathbf{T}$ is \emph{prepared} if it satisfies the above five conditions. To get the table prepared, the learner can perform the following operations:
\begin{description}[itemsep=5pt, align=left, leftmargin=0pt, font=\normalfont\itshape]
\item[Making \bm{$\mathbf{T}$} closed.] If $\mathbf{T}$ is not closed, there exists $r \in \bm{R}$ such that for all $s\in\bm{S}$ $\mathit{row}(r) \neq \mathit{row}(s)$. The learner thus can move such $r$ from $\bm{R}$ to $\bm{S}$. Moreover, each reset-logical-timed word $\pi(\Pi_{\{1,2\}}r \cdot \bm{\sigma})$ needs to be added to $\bm{R}$, where $\bm{\sigma}=(\sigma, 0)$ for all $\sigma \in \Sigma$. Such an operation is important since it guarantees that at every location all actions in $\Sigma$ are enabled, while specifying a clock valuation of these actions, despite that some invalid logical-timed words might be involved. Particularly, giving a bottom value $0$ as the clock valuation satisfies the precondition of the partition functions that will be described in Sect.~\ref{sbsc:hypo}.
\item[Making \bm{$\mathbf{T}$} consistent.] If $\mathbf{T}$ is not consistent, one inconsistency is resolved by adding $\bm{\sigma} \cdot e$ to $\bm{E}$, where $\bm{\sigma}$ and $e$ can be determined as follows. $T$ being inconsistent implies that 
there exist two reset-logical-timed words ${\gamma_r},{\gamma_r}'\in\bm{S}\cup\bm{R}$ at least, such that 
 ${\gamma_r} \cdot {\bm{\sigma_r}}, {\gamma_r}' \cdot {\bm{\sigma_r}}' \in\bm{S}\cup\bm{R}$ and $\Pi_{\{1,2\}}{\bm{\sigma_r}}=\Pi_{\{1,2\}}{\bm{\sigma_r}}'$ for some  ${\bm{\sigma_r}},{\bm{\sigma_r}}'\in\bm{\Sigma_r}$, with $\mathit{row}({\gamma_r})=\mathit{row}({\gamma_r}')$ but $\mathit{row}({\gamma_r}\cdot{\bm{\sigma_r}}) \neq \mathit{row}({\gamma_r}'\cdot {\bm{\sigma_r}}')$. So,   
let $\bm{\sigma} = \Pi_{\{1,2\}}{\bm{\sigma_r}} = \Pi_{\{1,2\}}{\bm{\sigma_r}}'$ and $e \in \bm{E}$ such that $f({\gamma_r} {\bm{\sigma_r}} \cdot e) \neq f({\gamma_r}' {\bm{\sigma_r}}' \cdot e)$. Thereafter, the learner fills the table by making membership queries. Note that this operation keeps the set $\bm{E}$ of suffixes being a set of logical-timed words.
\item[Making \bm{$\mathbf{T}$} evidence-closed.] If $\mathbf{T}$ is not evidence-closed, then the learner needs to add all prefixes of $\pi(\Pi_{\{1,2\}}s \cdot e)$ to $\bm{R}$ for every $s\in\bm{S}$ and $e\in\bm{E}$, except those already in $\bm{S}\cup\bm{R}$. Similarly, the learner needs to fill the table through membership queries.
%
\oomit{
\item[Making \bm{$\mathbf{T}$} prefix-closed.] If $\mathbf{T}$ is not prefix-closed, the learner should add all the prefixes of $\gamma_r\in\bm{S}\cup\bm{R}$ to $\bm{R}$ except those already in $\bm{S}\cup\bm{R}$. Similarly, the learner needs to fill the table through membership queries.
}
\end{description}

The condition that a timed observation table $\mathbf{T}$ is reduced and prefix-closed is inherently preserved by the aforementioned operations, together with the counterexample processing described later in Sect.~\ref{sbsc:eq_ctx}. Furthermore, a table may need several rounds of these operations before being prepared (cf. Algorithm~\ref{alg:learning}), since certain conditions may be violated by different, interleaved operations.

\subsection{Hypothesis construction}\label{sbsc:hypo}
As soon as the timed observation table $\mathbf{T}$ is prepared, a hypothesis can be constructed in two steps, i.e., the learner first builds a DFA $\text{M}$ based on the information in $\mathbf{T}$, and then transforms $\text{M}$ to a hypothesis $\mathcal{H}$, which will later be shown as a COTA.

Given a prepared timed observation table $\mathbf{T} = (\Sigma, \bm{\Sigma}, \bm{\Sigma_r}, \bm{S}, \bm{R}, \bm{E}, f)$, a DFA $\text{M} = (Q_{M}, \Sigma_{M}, \Delta_{M}, q^0_M, F_{M})$ can be built as follows:
\begin{itemize}
	\item the finite set of locations $Q_{M} = \{q_{\mathit{row}(s)} \mid s\in\bm{S}\}$;
	\item the initial location $q^0_M = q_{\mathit{row}(\epsilon)}$ for $\epsilon \in \bm{S}$;
	\item the set of accepting locations $F_{M} = \{q_{\mathit{row}(s)} \mid f(s\cdot \epsilon) = + \text{ for } s\in \bm{S} \text{ and } \epsilon \in \bm{E} \}$;
	\item the finite alphabet $\Sigma_{M} = \{\bm{\sigma_r} \in \bm{\Sigma_r} \mid \gamma_r \cdot \bm{\sigma_r} \in \bm{S}\cup\bm{R} \text{ for } \gamma_r \in \bm{\Sigma_r}^{*}\}$;
	\item the finite set of transitions $\Delta_M = \{(q_{\mathit{row}(\gamma_r)}, \bm{\sigma_r}, q_{\mathit{row}(\gamma_r\cdot \bm{\sigma_r})}) \mid \gamma_r \cdot \bm{\sigma_r} \in \bm{S}\cup\bm{R} \text{ for } \gamma_r \in \bm{\Sigma_r}^{*} \text{ and } \bm{\sigma_r} \in \bm{\Sigma_r}\}$.
\end{itemize}

The constructed DFA $\text{M}$ is compatible with the timed observation table $\mathbf{T}$ in the sense captured by the following lemma. 

\begin{lemma}\label{lemma:M=T}
	For a prepared timed observation table $\mathbf{T} = (\Sigma, \bm{\Sigma}, \bm{\Sigma_r}, \bm{S}, \bm{R}, \bm{E}, f)$, for every $\gamma_r\cdot e \in (\bm{S}\cup\bm{R})\cdot\bm{E}$, the constructed DFA $\text{\rm M} = (Q_{M}, \Sigma_{M}, \Delta_{M}, q^0_M, F_{M})$ accepts $\pi(\Pi_{\{1,2\}}\gamma_r\cdot e)$
	 if and only if $f(\gamma_r\cdot e) = +$. 
\end{lemma}

The learner then transforms the DFA $\text{M}$ to a hypothesis $\mathcal{H}=(\Sigma,Q,q_0,F,c,\Delta)$, with $Q = Q_M$, $q_0=q^0_M$, $F=F_M$, $c$ being the clock and $\Sigma$ the given alphabet as in $\mathbf{T}$. The set of transitions $\Delta$ in $\mathcal{H}$ can be constructed as follows: 
%
For any $q \in Q_{M}$ and $\sigma \in \Sigma$, let $\Psi_{q,\sigma} = \{ \mu \mid (q, (\sigma, \mu, b), q') \in \Delta_{M} \}$, 
then applying the partition function $P^c(\cdot)$ (defined below) to
  $\Psi_{q,\sigma}$ returns $k$ intervals, written as $I_1,\cdots, I_k$, satisfying $\mu_i \in I_i$ for any $1\leq i \leq k$, 
   where $k=|\Psi_{q,\sigma}|$; 
consequently, 
for every $(q, (\sigma, \mu_i, b_i), q') \in \Delta_{M}$, a fresh transition $\delta_i = (q,\sigma,I_i, b_i, q')$ is added to $\Delta$.


%
\begin{definition}[Partition function]\label{def:partition}
	Given a list of clock valuations $\ell =\mu_0,\mu_1,\cdots,\mu_n$ with 
	$0=\mu_0 < \mu_1 \cdots < \mu_n$, 
	and $\lfloor \mu_i \rfloor \neq \lfloor \mu_j \rfloor$ if $\mu_i,\mu_j\in\mathbb{R}_{\geq 0}\setminus\mathbb{N}$ and $i\neq j$ 
	for all $1\leq i, j\leq n$, let $\mu_{n+1} = 	 \infty$, then a partition function $P^c(\cdot)$ mapping $\ell$ to a set of intervals $\{I_0,I_1,\dots,I_n\}$, 
	which is a partition of $\mathbb{R}_{\geq 0}$, is defined as  
		\begin{equation*}
	I_i = \begin{cases}
	[\mu_i, \mu_{i+1}), & \text{if}\ \  \mu_i \in \mathbb{N} \wedge \mu_{i+1} \in \mathbb{N}; \\
	(\lfloor \mu_i \rfloor, \mu_{i+1}), & \text{if}\ \  \mu_i \in \mathbb{R}_{\geq 0}\setminus\mathbb{N} \wedge \mu_{i+1} \in \mathbb{N}; \\
	[\mu_i, \lfloor\mu_{i+1}\rfloor], & \text{if}\ \  \mu_i \in \mathbb{N} \wedge \mu_{i+1} \in \mathbb{R}_{\geq 0}\setminus\mathbb{N}; \\
	(\lfloor\mu_i\rfloor,\lfloor\mu_{i+1}\rfloor], & \text{if}\ \  \mu_i \in \mathbb{R}_{\geq 0}\setminus\mathbb{N} \wedge \mu_{i+1} \in \mathbb{R}_{\geq 0}\setminus\mathbb{N}.
	\end{cases}
	\end{equation*}
 \oomit{	\begin{itemize}
		\item $\bigcup_{0\leq i \leq n}{I_i} = [0,\infty)$;
		\item $I_i \cap I_j = \emptyset$ for $i\neq j$, $0\leq i, j \leq n$;
		\item $\mu_i \in I_i$ for $0\leq i \leq n$;
		\item for $0\leq i \leq n$, with $\mu_{n+1} = \infty$ and $\lfloor\mu_{n+1}\rfloor = \infty$,
		\begin{equation*}
		I_i = \begin{cases}
		[\mu_i, \mu_{i+1}), & \text{if}\ \  \mu_i \in \mathbb{N} \wedge \mu_{i+1} \in \mathbb{N}; \\
		(\lfloor \mu_i \rfloor, \mu_{i+1}), & \text{if}\ \  \mu_i \in \mathbb{R}_{\geq 0}\setminus\mathbb{N} \wedge \mu_{i+1} \in \mathbb{N}; \\
		[\mu_i, \lfloor\mu_{i+1}\rfloor], & \text{if}\ \  \mu_i \in \mathbb{N} \wedge \mu_{i+1} \in \mathbb{R}_{\geq 0}\setminus\mathbb{N}; \\
		(\lfloor\mu_i\rfloor,\lfloor\mu_{i+1}\rfloor], & \text{if}\ \  \mu_i \in \mathbb{R}_{\geq 0}\setminus\mathbb{N} \wedge \mu_{i+1} \in \mathbb{R}_{\geq 0}\setminus\mathbb{N}.
		\end{cases}
		\end{equation*}
	\end{itemize} }
\end{definition}
%
\begin{remark}
Definition~\ref{def:partition} is adapted from that in~\cite{Drews17} by imposing additional assumptions of the list of clock valuations in order to guarantee $\mu_i\in I_i$, for any $0\leq i\leq n$, due to the underlying continuous-time semantics. Whereas,
by $\mathbf{T}$ being prepared and
the normalization function described in Sect.~\ref{sbsc:eq_ctx}, the set of clock valuations $\Psi_{q,\sigma}$  can be arranged into a list $\ell_{q,\sigma} = \mu_0,\mu_1,\dots,\mu_n$ satisfying such assumptions given in Definition~\ref{def:partition}
for any $q\in Q_M$ and $\sigma\in \Sigma$. 
\end{remark}

\begin{example}\label{example:T_to_M_to_H}
Suppose $\mathbb{A}$ in Fig.~\ref{fig:dota_cota} recognizes the target timed language. Then the prepared table $\mathbf{T_5}$, the corresponding DFA $\text{M}_5$ and hypothesis $\mathcal{H}_5$ are depicted in Fig.~\ref{fig:T_to_M_to_H}. Here, the subscript $5$ indicates the fifth iteration of $\mathbf{T}$ (Details concerning the constructions and the entire learning process are enclosed in Appendix~\ref{appendix_learning_details}.
\end{example}
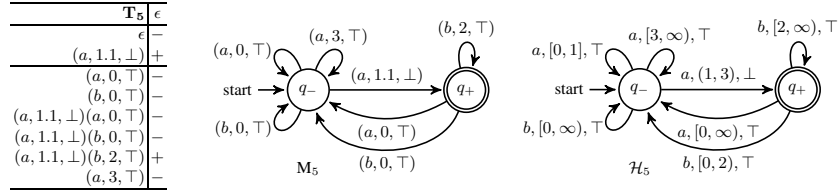
\begin{figure}[t]
	\hspace{.4cm}
	\begin{minipage}{0.2\linewidth}
		\begin{center}
			\resizebox{.9\textwidth}{!}{
				\begin{tabular}{r|c}
					\hline
					$\mathbf{T_5}$ & $\epsilon$ \\ 
					\hline
					$\epsilon$ & $-$ \\
					$(a,1.1,\bot)$ & $+$ \\
					\hline
					$(a,0,\top)$ & $-$ \\
					$(b,0,\top)$ & $-$ \\
					$(a,1.1,\bot)(a,0,\top)$ & $-$ \\
					$(a,1.1,\bot)(b,0,\top)$ & $-$ \\
					$(a,1.1,\bot)(b,2,\top)$ & $+$ \\
					$(a,3,\top)$ & $-$ \\
					\hline
				\end{tabular}
			}
		\end{center}			
	\end{minipage}
	\hspace{.2cm}
	\begin{minipage}{0.26\textwidth}
		\begin{center}
			\begin{tikzpicture}[->, >=stealth', shorten >=1pt, auto, node distance=3cm, semithick, scale=0.9, every node/.style={scale=0.7}]
			\node[initial,state] (0) {$q_{-}$};
			\node[accepting, state] (1) [right of = 0] {$q_{+}$};
			\path (0) edge [in= 145, out=110, loop] node[left] {$(a,0,\top)$} (0)
			(0) edge [in= 35, out=70, loop] node[above] {$(a,3,\top)$} (0)
			(0) edge [in= -145, out=-110, loop] node[left] {$(b,0,\top)$} (0)
			(0) edge node[above] {$(a,1.1,\bot)$} (1)
			(1) edge[loop above] node[above] {$(b,2,\top)$} (1)
			(1) edge[in=-30, out=-150] node[below] {$(a,0,\top)$} (0)
			(1) edge[in=-70, out=-110] node[below] {$(b,0,\top)$} (0);
			\node [below=25pt, align=flush center,text width=3cm] at (0) {$\text{M}_5$};
			\end{tikzpicture}
		\end{center}
	\end{minipage}
	\hspace{.8cm}
	\begin{minipage}{0.26\textwidth}
		\begin{center}
			\begin{tikzpicture}[->, >=stealth', shorten >=1pt, auto, node distance=3cm, semithick, scale=0.9, every node/.style={scale=0.7}]
			\node[initial,state] (0) {$q_{-}$};
			\node[accepting, state] (1) [right of = 0] {$q_{+}$};
			\path (0) edge [in= 145, out=110, loop] node[left] {$a,[0,1],\top$} (0)
			(0) edge [in= 35, out=70, loop] node[above] {$a,[3,\infty),\top$} (0)
			(0) edge [in= -145, out=-110, loop] node[left] {$b,[0,\infty),\top$} (0)
			(0) edge node[above] {$a,(1,3),\bot$} (1)
			(1) edge[loop above] node[above] {$b,[2,\infty),\top$} (1)
			(1) edge[in=-30, out=-150] node[below] {$a,[0,\infty),\top$} (0)
			(1) edge[in=-70, out=-110] node[below] {$b,[0,2),\top$} (0);
			\node [below=25pt, align=flush center,text width=3cm] at (0) {$\mathcal{H}_5$};
			\end{tikzpicture}
		\end{center}
	\end{minipage}
	\caption{The prepared timed observation table $\mathbf{T_5}$, the corresponding DFA $\text{M}_5$ and hypothesis $\mathcal{H}_5$.}
	\label{fig:T_to_M_to_H}
\end{figure}
\begin{lemma}\label{lemma:H=T}
	Given a DFA $\text{\rm M}= (Q_{M}, \Sigma_{M}, \delta_{M}, q^0_M, F_{M})$, which is generated from a prepared timed observation table $\mathbf{T}$, the hypothesis $\mathcal{H} = (\Sigma,Q,q_0,F,c,\Delta)$ is transformed from $\text{M}$. 
	For all $\gamma_r\cdot e \in (\bm{S}\cup\bm{R})\cdot\bm{E}$, 
	$\mathcal{H}$ accepts the reset-logical-timed word $\pi(\Pi_{\{1,2\}}\gamma_r \cdot e)$ iff $f(\gamma_r\cdot e) = +$. 
\end{lemma}

\begin{theorem}\label{theorem:H_is_cota}
	The hypothesis $\mathcal{H}$ is a COTA.
\end{theorem}

Given a clock valuation $\mu$, we denote the \emph{region} containing $\mu$ as $\llbracket \mu \rrbracket$, defined as $\llbracket \mu \rrbracket = [\mu,\mu]$ if $\mu\in\mathbb{N}$, and $\llbracket \mu \rrbracket = (\lfloor\mu\rfloor,\lfloor\mu\rfloor+1)$ otherwise. 
The following theorem establishes the compatibility of the constructed hypothesis $\mathcal{H}$ with the timed observation table $\mathbf{T}$.

\begin{theorem}\label{theorem:region_correct}
	For $\gamma_r \cdot e \in(\bm{S}\cup\bm{R})\cdot\bm{E}$, let $\pi(\Pi_{\{1,2\}} \gamma_r \cdot e) = (\sigma_1,\mu_1,b_1) \cdots (\sigma_n,\mu_n,b_n)$. Then for every $\mu_i'\in \llbracket \mu_i \rrbracket$, the hypothesis $\mathcal{H}$ accepts the reset-logical-timed word $\gamma_r'=(\sigma_1,\mu_1',b_1)\cdots(\sigma_n,\mu_n',b_n)$ if $f(\gamma_r \cdot e) = +$, and cannot accept it if $f(\gamma_r \cdot e) = -$. 
\end{theorem}


\subsection{Equivalence query and counterexample processing}\label{sbsc:eq_ctx}
%
Suppose that the teacher knows a COTA $\mathbb{A}$ which recognizes the target timed language $\mathcal{L}$. Then to answer an equivalence query is to determine whether $\mathcal{L}(\mathcal{H})=\mathcal{L}(\mathbb{A})$, which can be divided into two timed language inclusion problems, i.e., whether $\mathcal{L}(\mathcal{H})\subseteq\mathcal{L}(\mathbb{A})$ and $\mathcal{L}(\mathbb{A})\subseteq\mathcal{L}(\mathcal{H})$. Most decision procedures for language inclusion proceed by complementation and emptiness checking of the intersection~\cite{Hopcroft79}: $\mathcal{L}(A)\subseteq\mathcal{L}(B)$ iff $\mathcal{L}(A)\cap \overline{\mathcal{L}(B)} = \emptyset$. 
The fact that deterministic timed automata can be complemented~\cite{AlurM04} enables solving the inclusion problem by checking the emptiness of the resulted product automata $\mathcal{H}\times\overline{\mathbb{A}}$ and $\overline{\mathcal{H}}\times\mathbb{A}$. The complementation technique, however, does not apply to nondeterministic timed automata even if with only one single clock~\cite{Alur94}, which we plan to incorporate in our learning framework in future work.
We therefore opt for\footnote{Remark that the learning complexity (Sect.~\ref{sbsc:complexity}) is measured in terms of the number of queries rather than the time complexity of the specific method for checking the equivalence (nor membership). Additionally, the specific method of equivalence checking is not the main concern.} the alternative method presented by Ouaknine and Worrell in~\cite{Ouaknine04} showing that the language inclusion problem of timed automata with one clock (regardless of their determinacy) is decidable by reduction to a reachability problem on an infinite graph. 
That is, there exists a delay-timed word $\omega$ that leads to a \emph{bad configuration} if $\mathcal{L}(\mathcal{H})\nsubseteq\mathcal{L}(\mathbb{A})$. 
In detail, the corresponding run $\rho$ of $\omega$ ends in an accepting location in $\mathcal{H}$ but the counterpart $\rho'$ of $\omega$ in $\mathbb{A}$ is not accepting. Consequently, the teacher can provide the reset-delay-timed word $\omega_{r}$ resulted from $\omega$ as a negative counterexample $\mathit{ctx}_{-}$. Similarly, a positive counterexample $\mathit{ctx}_{+}=(\omega_{r},+)$ can be generated if $\mathcal{L}(\mathbb{A})\nsubseteq\mathcal{L}(\mathcal{H})$. An algorithm elaborating the equivalence query is provided in Appendix~\ref{appendix_algs}.

When receiving a counterexample $\mathit{ctx}=(\omega_{r},+/-)$, 
the learner first converts it to a reset-logical-timed word $\gamma_r= \Gamma(\omega_{r}) = (\sigma_1, \mu_1, b_1)(\sigma_2,\mu_2,b_2)\cdots(\sigma_n,\mu_n,b_n)$. 
By definition, $\gamma_r$ and $\omega_{r}$ share the same sequence of transitions in $\mathbb{A}$. Furthermore, by the contraposition of Theorem~\ref{theorem:L=>mathcalL}, $\gamma_r$ is an evidence for $L_r(\mathcal{H}) \neq L_r(\mathbb{A})$ if $\omega_{r}$ is an 
evidence for $\mathcal{L}(\mathcal{H}) \neq \mathcal{L}(\mathbb{A})$.

Additionally, by the definition of clock constraints $\Phi_c$, at any location, if an action $\sigma$ is enabled, i.e., its guard is satisfied, 
 w.r.t. the clock value $\mu \in \mathbb{R}_{\geq 0}\setminus\mathbb{N}$, then $\sigma$ should be enabled w.r.t. any clock value $\lfloor \mu \rfloor + \theta$ at the location, where $\theta \in (0,1)$. 
 Specifically, only one transition is available for $\sigma$ at the location on the interval  $\llbracket \mu \rrbracket$, because 
 the target automaton is deterministic. Therefore, in order to avoid unnecessarily distinguishing timed words and violating the assumptions of the list $\ell$ for the partition function, 
  the learner needs to apply a \emph{normalization function} $g$ to normalize $\gamma_r$. 

\begin{definition}[Normalization]
	A normalization function $g$ maps a reset-logical-timed word $\gamma_r= (\sigma_1, \mu_1, b_1)(\sigma_2,\mu_2,b_2)\cdots(\sigma_n,\mu_n,b_n)$ to another reset-logical-timed word by resetting any logical clock to its integer part plus a constant fractional part, i.e.,  $g(\gamma_r) = (\sigma_1, \mu_1', b_1)(\sigma_2,\mu_2',b_2)\cdots(\sigma_n,\mu_n',b_n)$, where $\mu_i' = \mu_i$ if $\mu_i \in \mathbb{N}$,  $\mu_i' = \lfloor \mu_i \rfloor + \theta$ for some fixed constant $\theta \in (0,1)$ otherwise. 
\end{definition}

We will instantiate $\theta = 0.1$ in what follows. Clearly our approach works for any other $\theta$ valued in $(0,1)$. 
This \emph{normalization} process guarantees the assumptions needed for Definition~\ref{def:partition}.

\oomit{The normalization function is used by the learner to resolve conflicts induced by the miss-distribution of reset-logical-timed actions in the generated DFA. That is, there may exist two actions $(\sigma,\mu,b)$ and $(\sigma,\mu',b')$ along two transitions from the same source location yet to different target locations, while sharing the same action $\sigma$ and different non-integer clock valuations $\mu \neq \mu'$ but $\lfloor \mu \rfloor = \lfloor \mu' \rfloor$. Such a miss-distribution blocks the learner from discovering the inconsistency of the underlying timed observation table $\mathbf{T}$, and furthermore violates the supposition of the partition function in Def.~\ref{def:partition}, which cannot be rectified. This is demonstrated by the following example.}

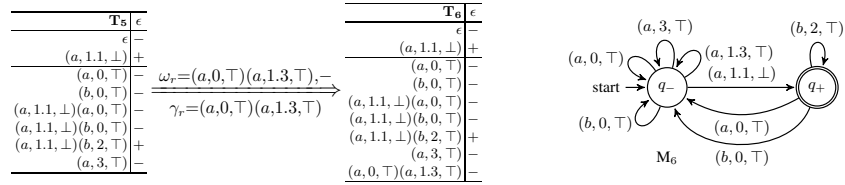
\begin{figure}[!t]
	\centering
	\hspace*{0.4cm}
	\begin{minipage}{0.52\textwidth}
		\centering
		\resizebox{\textwidth}{!}{
			\begin{tabular}{r|c}
				\hline
				$\mathbf{T_5}$ & $\epsilon$ \\ 
				\hline
				$\epsilon$ & $-$ \\
				$(a,1.1,\bot)$ & $+$ \\
				\hline
				$(a,0,\top)$ & $-$ \\
				$(b,0,\top)$ & $-$ \\
				$(a,1.1,\bot)(a,0,\top)$ & $-$ \\
				$(a,1.1,\bot)(b,0,\top)$ & $-$ \\
				$(a,1.1,\bot)(b,2,\top)$ & $+$ \\
				$(a,3,\top)$ & $-$ \\
				\hline
			\end{tabular}
			{\fontsize{15pt}{\linewidth}$\xLongrightarrow[\gamma_r=(a,0,\top)(a,1.3,\top)]{\omega_{r}=(a,0,\top)(a,1.3,\top), -}$}
			\begin{tabular}{r|c}
				\hline
				$\mathbf{T_6}$ & $\epsilon$ \\ 
				\hline
				$\epsilon$ & $-$ \\
				$(a,1.1,\bot)$ & $+$ \\
				\hline
				$(a,0,\top)$ & $-$ \\
				$(b,0,\top)$ & $-$ \\
				$(a,1.1,\bot)(a,0,\top)$ & $-$ \\
				$(a,1.1,\bot)(b,0,\top)$ & $-$ \\
				$(a,1.1,\bot)(b,2,\top)$ & $+$ \\
				$(a,3,\top)$ & $-$ \\
				$(a,0,\top)(a,1.3,\top)$ & $-$ \\
				\hline
			\end{tabular}
		}
	\end{minipage}
	\hfill
	\begin{minipage}{0.43\textwidth}
		\centering
        \hspace*{0.3cm}
		\resizebox{0.75\textwidth}{!}{
			\begin{tikzpicture}[->, >=stealth', shorten >=1pt, auto, node distance=3cm, semithick, scale=0.9, every node/.style={scale=0.9}]
			\node[initial,state] (0) {$q_{-}$};
			\node[accepting, state] (1) [right of = 0] {$q_{+}$};
			\path (0) edge [in= 160, out=125, loop] node[left] {$(a,0,\top)$} (0)
			(0) edge [in= 110, out=75, loop] node[above] {$(a,3,\top)$} (0)
			(0) edge [in= 65, out=30, loop] node[right] {$(a,1.3,\top)$} (0)
			(0) edge [in= -145, out=-110, loop] node[left] {$(b,0,\top)$} (0)
			(0) edge node[above] {$(a,1.1,\bot)$} (1)
			(1) edge[loop above] node[above] {$(b,2,\top)$} (1)
			(1) edge[in=-30, out=-150] node[below] {$(a,0,\top)$} (0)
			(1) edge[in=-70, out=-110] node[below] {$(b,0,\top)$} (0);
			\node [below=30pt, align=flush center,text width=3cm] at (0) {$\text{M}_6$};
			\end{tikzpicture}
		}
	\end{minipage}
	\caption{An illustration of the necessity of normalization by the normalization function.}
	\label{fig:wrong_ctx}
\end{figure}
\begin{example}\label{example:wrong_ctx}
	Consider the prepared table $\mathbf{T}_5$ in Fig.~\ref{fig:wrong_ctx} (as in Fig.~\ref{fig:T_to_M_to_H}). When the leaner asks 
	an equivalence query with hypothesis $\mathcal{H}_5$, the teacher answers that $\mathcal{L}(\mathcal{H}_5) \neq \mathcal{L}(\mathbb{A})$, where $\mathbb{A}$ in Fig.~\ref{fig:dota_cota} is the target automaton, and  provides a counterexample $(\omega_{r}, -)$ with $\omega_{r} =(a,0,\top)(a,1.3,\top)$, which can be transformed to a reset-logical-timed word $\gamma_r=(a,0,\top)(a,1.3,\top)$. If he adds prefixes of $\gamma_r$ to the table directly, the learner will get a prepared table $\mathbf{T}_6$ and thus construct a DFA $\text{M}_6$. Unfortunately, 
	the partition function defined in  Definition~\ref{def:partition} is not applicable to 
	 $(a,1.3,\top)$ and $(a,1.1,\bot)$ any more.  
	 On the other hand, if he adds the prefixes of the normalized reset-logical-timed word, i.e., $\gamma_r'=(a,0,\top)(a,1.1,\top)$, to $\mathbf{T}_5$, the learner will then get an inconsistent table whose consistency can be retrieved by the operation of ``making $\mathbf{T}$ consistent'' as expected.
\end{example}

The following theorem guarantees that the normalized reset-logical-timed word $\gamma_r'$ is also  an evidence for $L_r(\mathcal{H}) \neq L_r(\mathbb{A})$. Therefore, the learner can use it as a counterexample and thus adds all the prefixes of $\gamma_r'$ to $\bm{R}$ except those already in $\bm{S} \cup \bm{R}$.

\begin{theorem}\label{theorem:refine_correctness}
	Given a valid reset-logical-timed word $\gamma_r$ of $\mathbb{A}$, its normalization $\gamma_r'=g(\gamma_r)$
	shares the same sequence of transitions in $\mathbb{A}$.
\end{theorem}

\subsection{Learning algorithm}\label{sbsc:algorithm}

\begin{algorithm}[!t]
	\caption{Learning one-clock timed automaton with a smart teacher}
	\label{alg:learning}
	\SetKwInOut{Input}{input}
	\SetKwInOut{Output}{output}
	\Input{the timed observation table $\mathbf{T} = (\Sigma, \bm{\Sigma}, \bm{\Sigma_r}, \bm{S}, \bm{R}, \bm{E}, f)$.}
	\Output{the hypothesis $\mathcal{H}$ recognizing the target language $\mathcal{L}$.}
	$\bm{S}\leftarrow\{\epsilon\}$;
	$\bm{R}\leftarrow\{\Gamma(\omega) \mid \omega=(\sigma,0), \forall \sigma \in \Sigma \}$;
	$\bm{E}\leftarrow\{\epsilon\}$ \tcp*{initialization}
	fill $\mathbf{T}$ by membership queries\;
	$\mathit{equivalent}$ $\leftarrow$ $\bot$\;
	\While{$\mathit{equivalent}$ = $\bot$}{
		$prepared$ $\leftarrow$ is\_prepared($\mathbf{T}$) \tcp*{whether the table is prepared}
		\While{$prepared$ = $\bot$}
		{
			\lIf{$\mathbf{T}$ is not closed}{
				$\!$make\_closed($\mathbf{T}$)
			}
			\lIf{$\mathbf{T}$ is not consistent}{
				$\!$make\_consistent($\mathbf{T}$)
			}
			\lIf{$\mathbf{T}$ is not evidence-closed}{
				$\!$make\_evidence\_closed($\mathbf{T}$)
			}
			$prepared$ $\leftarrow$ is\_prepared($\mathbf{T}$)\;
		}
		$\text{M} \leftarrow$ build\_DFA($\mathbf{T}$) \tcp*{transforming $\mathbf{T}$ to a DFA $\text{M}$}
		$\mathcal{H} \leftarrow$ build\_hypothesis($\text{M}$) \tcp*{constructing a hypothesis $\mathcal{H}$ from $\text{M}$}
		$\mathit{equivalent}$, $\mathit{ctx}$ $\leftarrow$ equivalence\_query($\mathcal{H}$)\;
		\If{$\mathit{equivalent}$ = $\bot$}{
			ctx\_processing($\mathbf{T}$, $\mathit{ctx}$) \tcp*{counterexample processing}
		}
	}
	\Return $\mathcal{H}$\;
\end{algorithm}

We present in Algorithm~\ref{alg:learning} the learning procedure integrating all the previously stated ingredients, including preparing the table, membership and equivalence queries, hypothesis construction and counterexample processing. The learner first initializes the timed observation table $\mathbf{T}=(\Sigma, \bm{\Sigma}, \bm{\Sigma_r}, \bm{S}, \bm{R}, \bm{E}, f)$, where $\bm{S}=\{\epsilon\}$, $\bm{E}=\{\epsilon\}$, while for every $\sigma \in \Sigma$, he builds a logical-timed word $\gamma=(\sigma,0)$ and then obtains its reset counterpart  $\pi(\gamma)=(\sigma,0,b)$ by triggering a membership query to the teacher, which is then added to $\bm{R}$. Thereafter, the learner can fill the table by additional membership queries. Before constructing a hypothesis, the learner performs several rounds of operations described in Sect.~\ref{sbsc:membership} until $\mathbf{T}$ is prepared. Then, a hypothesis $\mathcal{H}$ is constructed leveraging an intermediate DFA $\text{M}$ and submitted to the teacher for an equivalence query. If the answer is positive, $\mathcal{H}$ recognizes the target language. Otherwise, the learner receives a counterexample $\mathit{ctx}$ and then conducts the counterexample processing to update $\mathbf{T}$ as described in Sect.~\ref{sbsc:eq_ctx}. The whole procedure repeats until the teacher gives a positive answer to an equivalence query.

To facilitate the analysis of correctness, termination and complexity of Algorithm~\ref{alg:learning}, we introduce the notion of \emph{symbolic state} that combines each location with its clock regions: a symbolic state of a COTA $\mathbb{A} = (\Sigma,Q,q_0,F,c,\Delta)$ is a pair $(q,\llbracket \mu \rrbracket)$, where $q\in Q$ and $\llbracket \mu \rrbracket$ is a region containing $\mu$. If $\kappa$ is the maximal constant appearing in the clock constraints of $\mathbb{A}$, then there exist $2\kappa+2$ such regions, including $[n,n]$ with $0\leq n\leq \kappa$, $(n,n+1)$ with $0\leq n< \kappa$, and $(\kappa, \infty)$ for each location, so there are a total of $\lvert Q \rvert \times (2\kappa+2)$ symbolic states. Then the correctness and termination of Algorithm~\ref{alg:learning} is stated in the following theorem, based on the fact that there is an injection from $\bm{S}$ (or equivalently, the locations of $\mathcal{H}$) to symbolic states of $\mathbb{A}$.
%

\begin{theorem} \label{theorem:termination}
  Algorithm~\ref{alg:learning} terminates and returns a COTA
  $\mathcal{H}$ which recognizes the target timed language
  $\mathcal{L}$.
\end{theorem}

\subsection{Complexity}\label{sbsc:complexity}
Given a target timed language $\mathcal{L}$ which is recognized by a
COTA $\mathbb{A}$, let $n = \lvert Q \rvert$ be the number of locations of
$\mathbb{A}$, $m = \lvert \Sigma \rvert$ the size of the alphabet, and $\kappa$ the
maximal constant appearing in the clock constraints of $\mathbb{A}$. In what follows, we derive the complexity of Algorithm~\ref{alg:learning} in terms of the number of queries.

By the proof of Theorem~\ref{theorem:termination}, $\mathcal{H}$  has at most $n(2\kappa+2)$
locations ( the size of $\bm{S}$) distinguished by $\bm{E}$. Thus,  $|\bm{E}|$ is
 at most $n(2\kappa+2)$
in order to distinguish these locations.
Therefore, the number of transitions of $\mathcal{H}$
is bounded by $mn^2(2\kappa+2)^3$. Furthermore, as every counterexample adds at least one fresh transition to the hypothesis
$\mathcal{H}$, where we consider each interval of the partition corresponds to 
a transition, this means that the number
of counterexamples and equivalence queries is at most $mn^2(2\kappa+2)^3$.

Now, we consider the number of membership queries, that is, to compute $(\lvert\bm{S}\rvert+\lvert\bm{R}\rvert)\times
\lvert\bm{E}\rvert$. Let $h$ be the
maximal length of counterexamples returned by the teacher, which is polynomial in the size of $\mathbb{A}$ according to 
 Theorem 5 in~\cite{VerwerWW11}, bounded by $n^2$. There are three cases of extending $\bm{R}$ by adding fresh rows, namely during the processing of counterexamples, making $\mathbf{T}$ closed, and making $\mathbf{T}$ evidence-closed. The first case adds at
most $hmn^2(2\kappa+2)^3$ rows to $\bm{R}$, while the latter two add at most $n(2\kappa+2)\times m$ and $n^2(2\kappa+2)^2$, respectively, yielding that the size of $\bm{R}$
is bounded by $\mathcal{O}(hmn^2\kappa^3)$, where $\mathcal{O}(\cdot)$ is the big Omicron notation. As a consequence, the number of membership queries is bounded by
$\mathcal{O}(mn^5\kappa^4)$.
So, the total complexity is $\mathcal{O}(mn^5\kappa^4)$. 

It is worth noting the above analysis is given in the worst case, where all partitions need to be fully refined. But, in practice we can learn the automaton without refining most partitions, and therefore the number of equivalence and membership queries, as well as the number of locations in the learned automaton are much fewer than the corresponding worst-case bounds. This will be demonstrated by examples in Sect.~\ref{sec:experiments}.

\subsection{Accelerating Trick}\label{sbsc:accelerating}
In the timed observation table, the function $f$ maps invalid reset-logical-timed words as well as certain valid ones to ``$-$'' when the teacher maintains a COTA $\mathbb{A}$ as the oracle. The learner thus needs multiple rounds of queries to distinguish the ``sink'' location from other unaccepting locations. 
If the function $f$ is extended to map invalid reset-logical-timed words to a distinct symbol, say ``$\times$'', when we let a DOTA $\mathcal{A}$ be the oracle, then the learner will take much fewer queries. We will later show in the experiments that such a trick significantly accelerates the learning process.

\section{Learning from a Normal Teacher}\label{sec:normalteacher}

In this section, we consider the problem of learning timed automata
with a normal teacher. As before, we assume the timed language to be
learned comes from a complete DOTA. For the normal teacher, inputs to
membership queries are delay-timed words, and the teacher returns
whether the word is in the language (without giving any additional
information). Inputs to equivalence queries are candidate DOTAs, and
the teacher either answers they are equivalent or provides a
delay-timed word as a counterexample.

The algorithm here is based on the procedure in the previous section.
We still maintain observation tables where the elements in $\bm{S}\cup\bm{R}$ are reset-logical-timed words and
the elements in $\bm{E}$ are logical-timed words. In order to obtain delay-timed words
for the membership queries, we need to \emph{guess} clock reset
information for transitions in the table. More precisely, in order to
convert a logical-timed word to a delay-timed word, it is necessary to
know clock reset information for all but the last transition. Hence,
it is necessary to guess reset information for each word in
$\bm{S}\cup \bm{R}$ (since $\bm{S}\cup \bm{R}$ is prefix-closed, this
is equivalent to guessing reset information for the last transition of
each word). Also, for each entry in
$(\bm{S}\cup \bm{R})\times \bm{E}$, it is necessary to guess all but
the last transition in $\bm{E}$. The algorithm can be thought of as
exploring a search tree, where branching is caused by guesses, and
successor nodes are constructed by the usual operations of preparing a
table and dealing with a counterexample.

\begin{algorithm}[!t]
	\caption{Learning one-clock timed automaton with a normal teacher}
	\label{alg:learning_normal}
	\SetKwInOut{Input}{input}
	\SetKwInOut{Output}{output}
	\Input{the timed observation table $\mathbf{T} = (\Sigma, \bm{\Sigma}, \bm{\Sigma_r}, \bm{S}, \bm{R}, \bm{E}, f)$.}
	\Output{the hypothesis $\mathcal{H}$ recognizing the target language $\mathcal{L}$.}
	$\mathit{ToExplore}\leftarrow\emptyset$;
	$\bm{S}\leftarrow\{\epsilon\}$;
	$\bm{R}\leftarrow\{\pi(\Gamma(\omega)) \mid \omega=(\sigma,0), \forall \sigma \in \Sigma \}$;
	$\bm{E}\leftarrow\{\epsilon\}$\; 
	$\mathit{currentTable} \leftarrow (\Sigma, \bm{\Sigma}, \bm{\Sigma_r}, \bm{S}, \bm{R}, \bm{E}, f)$\;
	$\mathit{tables}$ $\leftarrow$ guess\_and\_fill($\mathit{currentTable}$)\tcp*{guess resets and fill all table instances}
	$\mathit{ToExplore}$.insert($\mathit{tables}$)\tcp*{insert table instances $\mathit{tables}$ into $\mathit{ToExplore}$} 
	
	$\mathit{currentTable}$ $\leftarrow$ $\mathit{ToExplore}$.pop()\tcp*{pop out head instance as the current table}
	$\mathit{equivalent}$ $\leftarrow$ $\bot$\;
	\While{$\mathit{equivalent}$ = $\bot$}{
		$prepared$ $\leftarrow$ is\_prepared($\mathit{currentTable}$)\tcp*{whether the current table is prepared}
		\While{$prepared$ = $\bot$}
		{
			\If{$\mathit{currentTable}$ is not closed}{
				$\mathit{tables}$ $\leftarrow$ guess\_and\_make\_closed($\mathit{currentTable}$);
				$\mathit{ToExplore}$.insert($\mathit{tables}$)\;
				$\mathit{currentTable}$ $\leftarrow$ $\mathit{ToExplore}$.pop()\; 
			}
			\If{$\mathit{currentTable}$ is not consistent}{
				$\mathit{tables}$ $\leftarrow$ guess\_and\_make\_consistent($\mathit{currentTable}$);
				$\mathit{ToExplore}$.insert($\mathit{tables}$)\;
				$\mathit{currentTable}$ $\leftarrow$ $\mathit{ToExplore}$.pop()\; 
			}
			\If{$\mathit{currentTable}$ is not evidence-closed}{
				$\mathit{tables}$ $\leftarrow$ guess\_and\_make\_evidence\_closed($\mathit{currentTable}$);
				$\mathit{ToExplore}$.insert($\mathit{tables}$)\;
				$\mathit{currentTable}$ $\leftarrow$ $\mathit{ToExplore}$.pop()\;
			}
			$prepared$ $\leftarrow$ is\_prepared($\mathit{currentTable}$)\;
		}
		$\text{M} \leftarrow$ build\_DFA($\mathit{currentTable}$) \tcp*{transforming $\mathit{currentTable}$ to a DFA $\text{M}$}
		$\mathcal{H} \leftarrow$ build\_hypothesis($\text{M}$) \tcp*{constructing a hypothesis $\mathcal{H}$ from $\text{M}$}
		$\mathit{equivalent}$, $\mathit{ctx}$ $\leftarrow$ equivalence\_query($\mathcal{H}$)\tcp*{$\mathit{ctx}$ is a delay-timed word}
		\If{$\mathit{equivalent}$ = $\bot$}{
			$\mathit{tables}$ $\leftarrow$ guess\_and\_ctx\_processing($\mathit{currentTable}$, $\mathit{ctx}$) \tcp*{counterexample processing}
			$\mathit{ToExplore}$.insert($\mathit{tables}$)\;
			$\mathit{currentTable}$ $\leftarrow$ $\mathit{ToExplore}$.pop()\;
		}
	}
	\Return $\mathcal{H}$\;
\end{algorithm}

The detailed process is given in Algorithm
\ref{alg:learning_normal}. The learner maintains a set of table
instances, named $\mathit{ToExplore}$, which contains all table
instances that need to be explored.

The initial tables in $\mathit{ToExplore}$ are as follows. Each table
has $\bm{S}=\bm{E}=\{\epsilon\}$. For each $\sigma\in\Sigma$, there is
one row in $\bm{R}$ corresponding to the logical-timed word
$\omega=(\sigma, 0)$. It is necessary to guess a reset $b$ for each
$\omega$ thereby transforming it to a reset-logical-timed word
$\gamma_r=(\sigma, 0, b)$. There are $2^{\lvert \Sigma \rvert}$
possible combinations of guesses. These tables are filled by making
membership queries (in this case, the membership queries for each
table are the same). The resulting $2^{\lvert \Sigma \rvert}$ tables
form the initial tables in $\mathit{ToExplore}$.

In each iteration of the algorithm, one table instance is taken out of
$\mathit{ToExplore}$. The learner checks whether the table is closed,
consistent, and evidence closed. If the table is not closed,
i.e. there exists $r\in\bm{R}$ such that $row(r)\neq row(s)$ for all
$s\in \bm{S}$, the learner moves $r$ from $\bm{R}$ to $\bm{S}$. Then
for each $\sigma\in \Sigma$, the element $r\cdot (\sigma,0)$ is added
to $\bm{R}$, and a guess has to be made for its reset
information. Hence, $2^{\lvert \Sigma \rvert}$ unfilled table
instances will be generated. Next, for each entry in the
$\lvert\Sigma\rvert$ new rows of $\bm{R}$, it is necessary to guess
reset information for all but the last transition in
$e\in\bm{E}$. After this guess, it is now possible to fill the table
instances by making membership queries with transformed delay-timed words. Hence, there are at most
$2^{(\sum_{e_i \in \bm{E}\setminus\{\epsilon\}}{(\lvert e_i \rvert -
    1)}) \times \lvert \Sigma \rvert}$
filled table instances for one unfilled table instance. All filled
table instances are inserted into $\mathit{ToExplore}$.

If the table is not consistent, i.e. there exist some
$\gamma_r,\gamma_r'\in\bm{S}\cup\bm{R}$ and
$\bm{\sigma_r}\in\bm{\Sigma_r}$ such that
$\gamma_r\cdot\bm{\sigma_r},\gamma_r'\cdot\bm{\sigma_r}\in\bm{S}\cup\bm{R}$
and $row(\gamma_r)=row(\gamma_r')$, but
$row(\gamma_r\cdot\bm{\sigma_r}) \neq
row(\gamma_r'\cdot\bm{\sigma_r})$.
Let $e\in\bm{E}$ be one place where they are different. Then
$\bm{\sigma_r}\cdot e$ needs to be added to $\bm{E}$. For each entry in
$\bm{S}\cup\bm{R}$, all but the last transition in
$\bm{\sigma_r}\cdot e$ need to be guessed, then the table can be
filled.
$2^{(\lvert \bm{\sigma} \cdot e \rvert - 1) \times (\lvert \bm{S}
  \rvert + \lvert \bm{R} \rvert)}$
filled table instances will be generated and inserted into
$\mathit{ToExplore}$. The operation for making tables evidence-closed
is analogous.

Once the current table is prepared, the learner builds a hypothesis
$\mathcal{H}$ and makes an equivalence query to the teacher. If the
answer is positive, then $\mathcal{H}$ is a COTA which recognizes the
target timed language $\mathcal{L}$; otherwise, the teacher gives a
delay-timed word $\omega$ as a counterexample. The learner first finds
the longest reset-logical-timed word in $\bm{R}$ which, when converted to a
delay-timed word, agrees with a prefix of $\omega$. The remainder of
$\omega$, however, needs to be converted to a reset-logical-timed word by
guessing reset information. The corresponding prefixes are then added
to $\bm{R}$. Hence, at most $2^{\lvert\omega\rvert}$ unfilled table
instances are generated. For each unfilled table instance, at most
$2^{(\sum_{e_i \in \bm{E}\setminus\{\epsilon\}}{(\lvert e_i \rvert -
    1)})\times\lvert \omega \rvert}$
filled tables are produced and inserted into
$\mathit{ToExplore}$.

Throughout the learning process, the learner adds a finite number of
table instances to $\mathit{ToExplore}$ at every iteration. Hence, the
search tree is finite-branching. Moreover, if all guesses are correct,
the resulting table instance will be identical to the observation
table in the learning process with a smart teacher (apart from the
guessing processes, the basic table operations are the same as those
in Section~\ref{sbsc:membership}). This means, with an
appropriate search order, for example, taking the table instance that
requires the least number of guesses in $\mathit{ToExplore}$ at every
iteration, the algorithm terminates and returns the same table as in
the learning process with a smart teacher, which is a COTA that
recognizes the target language $\mathcal{L}$. In conformity to Theorem~\ref{theorem:L=>mathcalL}, the algorithm
may terminate even if the corresponding reset-logical-timed languages
are not equivalent, while yielding correct COTAs recognizing the same
delay-timed language.

\begin{theorem} \label{theorem:normal_correct}
  Algorithm~\ref{alg:learning_normal} terminates and returns a COTA
  $\mathcal{H}$ which recognizes the target timed language
  $\mathcal{L}$.
\end{theorem}

\paragraph{Complexity analysis.}

If
$\mathbf{T}=(\Sigma, \bm{\Sigma}, \bm{\Sigma_r}, \bm{S}, \bm{R},
\bm{E}, f)$
is the final observation table for the correct candidate COTA, the
number of guessed resets in $\bm{S}\cup\bm{R}$ is
$\lvert\bm{S}\rvert + \lvert\bm{R}\rvert$, and the number of guessed
resets for entries in each row of the table is
$\sum_{e_i \in \bm{E}\setminus\{\epsilon\}}{(\lvert e_i \rvert -1)}$.
Hence, the total number of guessed resets is
$(\lvert \bm{S} \rvert + \lvert \bm{R} \rvert)\times(1+\sum_{e_i \in
  \bm{E}\setminus\{\epsilon\}}{(\lvert e_i \rvert-1)})$.
Assuming an appropriate search order (for example according to the
number of guesses in each table), this yields the number of table
instances considered before termination as
$\mathcal{O}(2^{(\lvert \bm{S} \rvert + \lvert \bm{R} \rvert)\times(1+\sum_{e_i \in
  \bm{E}\setminus\{\epsilon\}}{(\lvert e_i \rvert-1)})})$.

\oomit{
out of each table. For every hypothesis $\mathcal{H}$, the learner makes an equivalence query to the teacher. If the answer is positive, then $\mathcal{H}$ is a COTA which recognizes the target timed language $\mathcal{L}$; otherwise, the teacher gives a delay-timed word $\omega'$ as a counterexample, which is then transformed to $2^{\lvert \omega' \rvert}$ reset-logical-timed words by again guessing the resets. The learner hence needs a copy of the corresponding table $2^{\lvert \omega' \rvert}$ times to address all possible situations. However, there could be a number of conflicts when he adds all the prefixes of a normalized reset-logical-timed word to a copy, i.e., a prefix $\gamma_r$ is different from an element $\gamma_r'$ in $\bm{S}\cup\bm{R}$ yet $\Pi_{\{1,2\}}\gamma_r = \Pi_{\{1,2\}}\gamma_r'$. The learner may simply drop such a conflicting table and get a non-conflicting one prepared. In particular, the learner needs to guess the resets of elements in $\bm{E}$ when making a table evidence closed, yielding again exponential copies, that is, if the longest element in the current $\bm{E}$ is of the length $l$, then the number of copies is $\lvert \bm{E} \rvert \times 2^l$. The learner may also drop a table whenever it violates the supposition of the partition function. Throughout the learning process, the learner always keeps a table identical to the timed table in the learning process with a smart teacher. The process thus terminates and returns a COTA which recognizes the target language $\mathcal{L}$. It might however terminate even if the corresponding reset-logical-timed languages are not equivalent, while yielding correct COTAs recognizing the same timed language. A heuristic for selecting tables is to always deal with the table featuring the smallest size of $\bm{S}$, which often leads to a COTA of smaller size. In brief, the learner can still learn a correct DOTA by querying with delay-timed words while assembling clock resets by simple guesses, yet with an exponential blow up in the complexity compared with that of a smart teacher.}

\vspace*{-0.2cm}
\section{Implementation and Experimental Results}\label{sec:experiments}
\vspace*{-0.2cm}

To investigate the efficiency and scalability of the proposed methods,
we implemented a prototype\footnote{Available
  at~\url{https://github.com/Leslieaj/OTALearning}. The evaluated artifact is archived in~\cite{figshare}.} in
\textsc{Python} for learning deterministic one-clock timed
automata. The examples include a practical case concerning the functional specification of the TCP
protocol~\cite{rfc793} and a set of randomly generated DOTAs to be
learnt.
All of the evaluations have been carried
out on a 3.6GHz Intel Core-i7 processor with 8GB RAM running 64-bit
Ubuntu 16.04.

\paragraph{Functional specification of the TCP protocol.}
\vspace*{-0.2cm}

In~\cite{rfc793}, a state diagram on page 23 specifies state changes
during a TCP connection triggered by causing events while leading to
resulting actions.
As observed by Ouaknine and Worrell in~\cite{Ouaknine04}, such a
functional specification of the protocol can be represented as a
one-clock timed automaton. In our setting, the corresponding DOTA
$\mathcal{A}$ to be learnt is configured to have
$\lvert Q \rvert = 11$ states with the two $\mathtt{CLOSED}$ states
collapsed, $\lvert \Sigma \rvert = 10$ after abstracting the causing
events and the resulting actions, $\lvert F \rvert = 2$, and
$\lvert \Delta \rvert = 19$ with appropriately specified timing
constraints including guards and resets. Using the algorithm with the
smart teacher, a correct DOTA $\mathcal{H}$ is learned in 155 seconds
after 2600 membership queries and 28 equivalence
queries. Specifically, $\mathcal{H}$ has 15 locations excluding a sink
location connected by 28 transitions. The introduction of 4 new
locations comes from splitting of guards along transitions, which
however can be trivially merged back with other locations. The figures
depicting $\mathcal{A}$ and $\mathcal{H}$ can be found in
Appendix~\ref{appendix_tcp}.
%

\paragraph{Random examples for a smart teacher.}

We randomly generated 80 DOTAs in eight groups, with each group having
different numbers of locations, size of alphabet, and maximum constant
appearing in clock constraints. As shown in
Table~\ref{tb:random_experiment}, the proposed learning method
succeeds in all cases in identifying a DOTA that recognizes the same
timed language. In particular, the number of membership queries and
that of equivalence queries appear to grow polynomially with the size
of the problem\footnote{An exception w.r.t. the group 7\_6\_10 is due
  to relatively simple DOTAs generated occasionally.}, and are much
smaller than the worst-case bounds estimated in
Sect.~\ref{sbsc:complexity}. Moreover, the learned DOTAs do not have
prominent increases in the number of locations (by comparing
$n_{\text{mean}}$ with the first component of Case IDs). The average
wall-clock time including both time taken by the learner and by the
teacher is recorded in the last column $t_{\text{mean}}$, of which,
however, often over 90\% is spent by the teacher for checking
equivalences w.r.t. small $\mathbf{T}$'s while around 50\% by the
learner for checking the preparedness condition w.r.t. large
$\mathbf{T}$'s.

\begin{table}[t]
	\vspace*{-.3cm}
	\caption{Experimental results on random examples for the smart teacher situation.}
	\label{tb:random_experiment}
	\vspace*{-.3cm}
	\begin{center}
	\begin{tabular}{rcrcP{25pt}P{28pt}P{23pt}cP{20pt}P{23pt}P{18pt}c*{2}{P{25pt}}}
		\toprule
		\multirow{2}*{Case ID} & & \multirow{2}*{$\lvert\Delta\rvert_{\text{mean}}$} & &  \multicolumn{3}{c}{{\#}Membership} & & \multicolumn{3}{c}{{\#}Equivalence} & & \multirow{2}*{$n_{\text{mean}}$} & \multirow{2}*{$t_{\text{mean}}$}\\
		\cmidrule{5-7} \cmidrule{9-11}
		& & & & $N_{\text{min}}$ & $N_{\text{mean}}$ & $N_{\text{max}}$ & & $N_{\text{min}}$ & $N_{\text{mean}}$ & $N_{\text{max}}$&  &  & \\
		\midrule
		4\_4\_20 & & 16.3 & & 118 & 245.0 & 650 & & 20 & 30.1 & 42 & & 4.5 & 24.7\\
		7\_2\_10 & & 16.9 & & 568 & 920.8 & 1393 & & 23 & 31.3 & 37 & & 9.1 & 14.6\\
		7\_4\_10 & & 25.7 & & 348 & 921.7 & 1296 & & 34 & 50.9 & 64 & & 9.3 & 38.0 \\
		7\_6\_10 & & 26.0 & & 351 & 634.5 & 1050 & & 35 & 44.7 & 70 & & 7.8 & 49.6\\
		7\_4\_20 & & 34.3 & & 411 & 1183.4 & 1890 & & 52 & 70.5 & 93 & & 9.5 & 101.7\\
		10\_4\_20 & & 39.1 & & 920 & 1580.9 & 2160 & & 61 & 73.1 & 88 & & 11.7 & 186.7\\
		12\_4\_20 & & 47.6 & & 1090 & 2731.6 & 5733 & & 66 & 97.4 & 125 & & 16.0 & 521.8\\
		14\_4\_20 & & 58.4 & & 1390 & 2238.6 & 4430 & & 79 & 107.7 & 135 & & 16.0 & 515.5\\
		\bottomrule
	\end{tabular}
	\end{center}
\vspace*{-2mm}
\begin{minipage}{\linewidth}
{\fontsize{8pt}{\baselineskip}
    Case ID: $n$\_$m$\_$\kappa$, consisting of the number of locations, the size of the alphabet and the maximum constant appearing in the clock constraints, respectively, of the corresponding group of $\mathcal{A}$'s.\\
	$\lvert\Delta\rvert_{\text{mean}}$: the average number of transitions in the corresponding group.\\
	{\#}Membership \& {\#}Equivalence: the number of conducted membership and equivalence queries, respectively. $N_{\text{min}}$: the minimal, $N_{\text{mean}}$: the mean, $N_{\text{max}}$: the maximum.\\
	$n_{\text{mean}}$: the average number of locations of the learned automata in the corresponding group.\\
	$t_{\text{mean}}$: the average wall-clock time in seconds, including that taken by the learner and the teacher.
}
\end{minipage}
\end{table}

It is worth noting that all of the results reported above are carried out on an implementation equipped with the \emph{accelerating trick} discussed in Sect.~\ref{sbsc:accelerating}. We remark that when \emph{dropping} this trick, the average number of membership queries blow up with a factor of 0.83 (min) to 15.02 (max) with 2.16 in average for all the 8 groups, and 0.84 (min) to 1.71 (max) with 1.04 for the average number of equivalence queries, leading to dramatic increases also in the computation time (including that in operating tables). The alternative implementation and experimental results without the accelerating trick can also be found in the tool page (under the \texttt{dev} branch).

\begin{table}[t]
	\vspace*{-.3cm}
	\caption{Experimental results on random examples for the normal teacher situation.}
	\label{tb:random_experiment_normal}
	\vspace*{-.3cm}
	\begin{center}
		\begin{tabular}{rcrcP{21pt}P{21pt}P{21pt}cP{21pt}P{21pt}P{18pt}c*{2}{P{20pt}}P{35pt}P{33pt}}
			\toprule
			\multirow{2}*{Case ID} & & \multirow{2}*{$\lvert\Delta\rvert_{\text{mean}}$} & &  \multicolumn{3}{c}{{\#}Membership} & & \multicolumn{3}{c}{{\#}Equivalence} & & \multirow{2}*{$n_{\text{mean}}$} & \multirow{2}*{$t_{\text{mean}}$} & \multirow{2}*{{\#}$\mathbf{T}_{\textit{explored}}$} & \multirow{2}*{{\#}Learnt}\\
			\cmidrule{5-7} \cmidrule{9-11}
			& & & & $N_{\text{min}}$ & $N_{\text{mean}}$ & $N_{\text{max}}$ & & $N_{\text{min}}$ & $N_{\text{mean}}$ & $N_{\text{max}}$&  &  & &  & \\
			\midrule
			3\_2\_10 & & 4.8 & & 43 & 83.7 & 167 & & 5 & 8.8 & 14 & & 3.0 & 0.9 & 149.1 & 10/10 \\
			4\_2\_10 & & 6.8 & & 67 & 134.0 & 345 & & 6 & 13.3 & 24 & & 4.0 & 7.4 & 563.0 & 10/10 \\
			5\_2\_10 & & 8.8 & & 75 & 223.9 & 375 & & 9 & 15.2 & 24 & & 5.0 & 35.5 & 2811.6 & 10/10 \\
			6\_2\_10 & & 11.9 & & 73 & 348.3 & 708 & & 10 & 16.7 & 30 & & 5.6 & 59.8 & 5077.6 & 7/10 \\
			4\_4\_20 & & 16.3 & & 231 & 371.0 & 564 & & 27 & 30.9 & 40 & & 4.0 & 137.5 & 8590.0 & 6/10 \\
			\bottomrule
		\end{tabular}
	\end{center}
	\vspace*{-2mm}
	\begin{minipage}{\linewidth}
        {\fontsize{8pt}{\baselineskip}
			{\#}Membership \& {\#}Equivalence: the number of conducted membership and equivalence queries with the cached methods, respectively. $N_{\text{min}}$: the minimal, $N_{\text{mean}}$: the mean, $N_{\text{max}}$: the maximum.\\
			{\#}$\mathbf{T}_{\textit{explored}}$: the average number of the explored table instances.\\
			{\#}Learnt: the number of the learnt DOTAs in the group (learnt/total).
		}
	\end{minipage}
\end{table}

\paragraph*{Random examples for a normal teacher.}
Due to its high, exponential complexity, the algorithm with a normal teacher failed (out of memory) in identifying DOTAs for almost all the above examples, except 6 cases out of the 10 in group $4\_4\_20$. We therefore randomly
  generated 40 extra DOTAs of smaller size classified into 4 groups. With the accelerating trick, the
  learner need not guess the resets in elements of $\bm{E}$ for an
  entry in $\bm{S}\cup\bm{R}$ if the querying result of the entry is
  the sink location. We also omitted the checking of the evidence-closed
  condition, since it may add redundant rows in $\bm{R}$, leading to more guesses and thereby a larger search space. The omission does not affect the
  correctness of the learnt DOTAs. Moreover, as different table instances
  may generate repeated queries, we cached the results of membership
  queries and counterexamples, such that the numbers of membership and
  equivalence queries to the teacher can be significantly reduced. Table~\ref{tb:random_experiment_normal} shows the
  performance of the algorithm in this setting. Results without caching are available
  in the tool page (under the \texttt{normal} branch).

\section{Conclusion}\label{sec:conclusion}

We have presented a polynomial active learning method for
deterministic one-clock timed automata from a smart teacher who can
tell information about clock resets in membership and equivalence
queries. Our technique is based on converting the problem to that of
learning reset-logical-timed languages. 
We then extend the method to learning DOTAs from a normal teacher who receives delay-timed words
for membership queries, while the learner guesses the reset information in the
observation table. We evaluate both algorithms on randomly generated
examples and, for the former case, the functional
specification of the TCP protocol.

Moving forward, an extension of our active learning method to
nondeterministic OTAs and timed automata involving multiple
clocks is of particular interest.

	
\newpage
\bibliographystyle{splncs04}
\bibliography{reference}

\newpage

\begin{subappendices}
\renewcommand{\thesection}{\Alph{section}} 

\section{Proofs for Lemmas and Theorems}\label{appendix_proof}


\begin{proof}[of Theorem~\ref{theorem:L=>mathcalL}]
	By the definitions of delay-timed word and reset-delay-timed word, it suffices that $\mathcal{L}(\mathcal{A})=\mathcal{L}(\mathcal{B})$ if $\mathcal{L}_r(\mathcal{A})=\mathcal{L}_r(\mathcal{B})$. By the definitions of reset-delay-timed word and reset-logical-timed word with their mutual transforming method, we conclude that $\mathcal{L}_r(\mathcal{A})=\mathcal{L}_r(\mathcal{B})$ iff $L_r(\mathcal{A})=L_r(\mathcal{B})$. Hence, if $L_r(\mathcal{A})=L_r(\mathcal{B})$, then $\mathcal{L}(\mathcal{A})=\mathcal{L}(\mathcal{B})$. This completes the proof. \qed
\end{proof}

\begin{proof}[of Lemma~\ref{lemma:M=T}]
	A reset-logical-timed word $\gamma_r \in \bm{S}\cup\bm{R}$ happens in two cases, i.e., $\gamma_r\in\bm{S}$ or $\gamma_r\in\bm{R}$. For the first case, $\pi(\Pi_{\{1,2\}}\gamma_r \cdot e) \in \bm{S}\cup\bm{R}$ holds for all $e\in\bm{E}$ since $\mathbf{T}$ is evidence-closed. Hence let $\gamma_r' = \pi(\Pi_{\{1,2\}}\gamma_r \cdot e)$, then obviously $\gamma_r' \in \bm{S}\cup\bm{R}$. Therefore, if $f(\gamma_r\cdot e) = +$, then $f(\gamma_r'\cdot \epsilon)=+$, meaning that $\gamma_r'$ ends in $q_{\mathit{row}(\gamma_r')}\in F_{M}$, namely the constructed DFA $\text{M}$ accepts $\pi(\Pi_{\{1,2\}}\gamma_r \cdot e)$. Furthermore, if $f(\gamma_r \cdot e) = -$, then $f(\gamma_r'\cdot\epsilon)=-$, indicating that $\gamma_r'$ ends in $q_{\mathit{row}(\gamma_r')}\notin F_{M}$. This follows that $\text{M}$ does not accept $\pi(\Pi_{\{1,2\}}\gamma_r \cdot e)$. 
	
	For the second case, i.e. $\gamma_r\in\bm{R}$, then there exists $\gamma_r'\in\bm{S}$ such that $\mathit{row}(\gamma_r')=\mathit{row}(\gamma_r)$ since $\mathbf{T}$ is closed, which further implies that $f(\gamma_r \cdot e)=f(\gamma_r' \cdot e)$ for all $e\in\bm{E}$. Thus, it is reduced to the first case. 
	\qed
\end{proof}

\begin{proof}[of Lemma~\ref{lemma:H=T}]
	Given a DFA $\text{M}$, 
	by the above construction, 
	for each transition $(q,(\sigma,\mu, b),q')\in\Delta_{M}$, there is a corresponding transition $\delta=(q,\sigma,I,b, q')$ where $\mu \in I \in P^{c}{(\Psi_{q,\sigma})}$ in the hypothesis $\mathcal{H}$. Hence, given a reset-logical-timed word $\pi(\Pi_{\{1,2\}}\gamma_r \cdot e) = (\sigma_1,\mu_1, b_1)\cdots(\sigma_n,\mu_n,b_n)$, $\mathcal{H}$ accepts this word iff $\text{M}$ accepts it. By Lemma~\ref{lemma:M=T}, $\mathcal{H}$ accepts $\pi(\Pi_{\{1,2\}}\gamma_r \cdot e)$ iff $f(\gamma_r\cdot e) = +$. 
	\qed
\end{proof}

\begin{proof}[of Theorem~\ref{theorem:H_is_cota}]
	First, in order to guarantee $\mathbf{T}$ being closed when moving an element $r\in\bm{R}$ to $\bm{S}$, the learner adds the reset-logical-timed word $\pi(\Pi_{\{1,2\}}r \cdot (\sigma,0))$ for every $\sigma\in\Sigma$ to $\bm{R}$, which means that 
	there is always at least a outgoing transition from $q_{\mathit{row}(r)}$ for every action in $\Sigma$. 
	Secondly,  Definition~\ref{def:partition} implies 
	that $P^{c}{(\Psi_{q_{\mathit{row}(r)},\sigma})}$ is a partition of $\mathbb{R}_{\geq 0}$ for every $\sigma\in\Sigma$. Hence, $\mathcal{H}$ is a COTA. \qed
\end{proof}

\begin{proof}[of Theorem~\ref{theorem:region_correct}]
	By Lemma~\ref{lemma:M=T},~\ref{lemma:H=T} and Theorem~\ref{theorem:H_is_cota}, for every $\gamma_r \cdot e \in(\bm{S}\cup\bm{R})\cdot\bm{E}$, there exists a unique accepting run $\rho$ that admits $\pi(\Pi_{\{1,2\}}\gamma_r \cdot e)$ if $f(\gamma_r \cdot e) = +$. Hence, every logical-timed action $(\sigma_i,\mu_i,b_i)$, with $1 \leq i \leq n$, triggers a transition $\delta_i=(q_i, \sigma_i, \phi_i, b_i, q_{i+1})$ from $q_i$ to $q_{i+1}$, where $\mu_i \in \phi_i$ by Definition~\ref{def:partition}. By the above definition,
	$\llbracket \mu_i \rrbracket \subseteq \phi_i$, therefore $(\sigma_i,\mu_i',b_i)$ can also trigger the transition $\delta_i$. Hence, there exists a unique accepting run $\rho'$ that admits $\gamma_r'=(\sigma_1,\mu_1',b_1)\cdots(\sigma_n,\mu_n',b_n)$,
	i.e., $\mathcal{H}$ admits $\gamma_r'$. Suppose it is not the case when $f(\gamma_r\cdot e) = -$, 
	it is easy to follow  $f(\gamma_r\cdot e) = +$ by Lemma~\ref{lemma:H=T}, which contradicts to $f(\gamma_r\cdot e) = -$. \qed 
\end{proof}

\begin{proof}[of Theorem~\ref{theorem:refine_correctness}]
	Let $\gamma_r= (\sigma_1, \mu_1, b_1)(\sigma_2,\mu_2,b_2)\cdots(\sigma_n,\mu_n,b_n)$ and its normalization $\gamma_r' = (\sigma_1, \mu_1', b_1)(\sigma_2,\mu_2',b_2)\cdots(\sigma_n,\mu_n',b_n)$. By the definition of $g$, $\llbracket \mu' \rrbracket = \llbracket \mu \rrbracket$. Therefore, if $(\sigma_i, \mu_i, b_i)$ fires a transition $\delta_i=(q_i, \sigma_i, \phi_i, b_i, q_{i+1})$, 
	then $(\sigma_i, \mu_i', b_i)$ can also fire the same transition according to the definition of $\Phi_c$. 
	Specifically, by the assumption that $\mathbb{A}$ is deterministic, both the two timed actions can only 
	be taken by the transition. Hence, $\gamma_r'$ and $\gamma_r$ share the same sequence of transitions in $\mathbb{A}$. \qed 
\end{proof}
\begin{proof}[of Theorem~\ref{theorem:termination}]
	By Theorem~\ref{theorem:H_is_cota}, the returned hypothesis $\mathcal{H}$ is a COTA. Then the correctness (i.e. $\mathcal{H}$ recognizes the target timed language) follows directly from the equivalence query.
	Now we prove the termination. Observe that each reset-logical-timed word $s \in \bm{S}$ corresponds to a symbolic state reached after running $s$ on $\mathbb{A}$. Since $\mathbf{T}$ is reduced, implying that given any two elements $s_1,s_2 \in \bm{S}$, there exists $e\in E$ such that running $s_1\cdot e$ and $s_2\cdot e$ on $\mathbb{A}$ gives different
	acceptance results. Further by Theorem~\ref{theorem:region_correct}, $s_1$ and $s_2$ must reach different symbolic states of $\mathbb{A}$. Hence, there is an injection from $\bm{S}$ (or
	equivalently, the locations of $\mathcal{H}$) to symbolic states of $\mathbb{A}$. It follows that the size of the set $\bm{S}$ is bounded by $\lvert Q \rvert \times (2\kappa+2)$. Since each iteration of the algorithm either adds at least one element to $\bm{S}$ or refines at least one of the partitions along transitions of $\mathcal{H}$, or both, the algorithm is guaranteed to terminate. \qed
\end{proof}

\begin{proof}[of Theorem~\ref{theorem:normal_correct}]
   The algorithm can be viewed as a breadth-first-search (BFS) on
    a finite multi-way tree, each of whose nodes is a filled table
    instance. The depth of a node is the number of guessed resets. In
    other words, table instances at the same depth have the same
    number of guesses. The learner takes out a table instance that
    required the least number of guesses in $\mathit{ToExplore}$ at
    every iteration. If
    $\mathbf{T} = (\Sigma, \bm{\Sigma}, \bm{\Sigma_r}, \bm{S}, \bm{R},
    \bm{E}, f)$
    is the final table in the learning process with a smart teacher,
    $\mathbf{T}$ can be found at depth at most
    $(\lvert \bm{S} \rvert + \lvert \bm{R} \rvert)\times(1+\sum_{e_i
      \in \bm{E}\setminus\{\epsilon\}}{(\lvert e_i \rvert-1)})$
    of the tree, since it corresponds to choosing the correct table
    instance corresponding to the smart teacher situation at every
    guess. Then the learner can find $\mathbf{T}$ after
    $2^{(\lvert \bm{S} \rvert + \lvert \bm{R}
      \rvert)\times(1+\sum_{e_i \in
        \bm{E}\setminus\{\epsilon\}}{(\lvert e_i \rvert-1)})}$
    steps in the worst case, since the learner has to check all of the
    tables at earlier depths before entering depth
    $(\lvert \bm{S} \rvert + \lvert \bm{R} \rvert)\times(1+\sum_{e_i
      \in \bm{E}\setminus\{\epsilon\}}{(\lvert e_i \rvert-1)})$
    and finds $\mathbf{T}$. Consequently, the algorithm terminates and
    returns a correct COTA $\mathcal{H}$ which recognizes the target
    timed language $\mathcal{L}$. \qed
\end{proof}

\newpage
\section{Detailed Learning Process for the DOTA $\mathcal{A}$ in Fig.~\ref{fig:dota_cota}}\label{appendix_learning_details}

\begin{figure}[h]
	\hspace{.4cm}
	\begin{minipage}{0.2\linewidth}
		\begin{center}
			\resizebox{.9\textwidth}{!}{
				\begin{tabular}{r|c}
					\hline
					$\mathbf{T_5}$ & $\epsilon$ \\ 
					\hline
					$\epsilon$ & $-$ \\
					$(a,1.1,\bot)$ & $+$ \\
					\hline
					$(a,0,\top)$ & $-$ \\
					$(b,0,\top)$ & $-$ \\
					$(a,1.1,\bot)(a,0,\top)$ & $-$ \\
					$(a,1.1,\bot)(b,0,\top)$ & $-$ \\
					$(a,1.1,\bot)(b,2,\top)$ & $+$ \\
					$(a,3,\top)$ & $-$ \\
					\hline
				\end{tabular}
			}
		\end{center}			
	\end{minipage}
	\hspace{.2cm}
	\begin{minipage}{0.26\textwidth}
		\begin{center}
			\begin{tikzpicture}[->, >=stealth', shorten >=1pt, auto, node distance=3cm, semithick, scale=0.9, every node/.style={scale=0.7}]
			\node[initial,state] (0) {$q_{-}$};
			\node[accepting, state] (1) [right of = 0] {$q_{+}$};
			\path (0) edge [in= 145, out=110, loop] node[left] {$(a,0,\top)$} (0)
			(0) edge [in= 35, out=70, loop] node[above] {$(a,3,\top)$} (0)
			(0) edge [in= -145, out=-110, loop] node[left] {$(b,0,\top)$} (0)
			(0) edge node[above] {$(a,1.1,\bot)$} (1)
			(1) edge[loop above] node[above] {$(b,2,\top)$} (1)
			(1) edge[in=-30, out=-150] node[below] {$(a,0,\top)$} (0)
			(1) edge[in=-70, out=-110] node[below] {$(b,0,\top)$} (0);
			\node [below=25pt, align=flush center,text width=3cm] at (0) {$\text{M}_5$};
			\end{tikzpicture}
		\end{center}
	\end{minipage}
	\hspace{.8cm}
	\begin{minipage}{0.26\textwidth}
		\begin{center}
			\begin{tikzpicture}[->, >=stealth', shorten >=1pt, auto, node distance=3cm, semithick, scale=0.9, every node/.style={scale=0.7}]
			\node[initial,state] (0) {$q_{-}$};
			\node[accepting, state] (1) [right of = 0] {$q_{+}$};
			\path (0) edge [in= 145, out=110, loop] node[left] {$a,[0,1],\top$} (0)
			(0) edge [in= 35, out=70, loop] node[above] {$a,[3,\infty),\top$} (0)
			(0) edge [in= -145, out=-110, loop] node[left] {$b,[0,\infty),\top$} (0)
			(0) edge node[above] {$a,(1,3),\bot$} (1)
			(1) edge[loop above] node[above] {$b,[2,\infty),\top$} (1)
			(1) edge[in=-30, out=-150] node[below] {$a,[0,\infty),\top$} (0)
			(1) edge[in=-70, out=-110] node[below] {$b,[0,2),\top$} (0);
			\node [below=25pt, align=flush center,text width=3cm] at (0) {$\mathcal{H}_5$};
			\end{tikzpicture}
		\end{center}
	\end{minipage}
	\caption{The prepared timed observation table $\mathbf{T_5}$, the corresponding DFA $\text{M}_5$ and hypothesis $\mathcal{H}_5$. A copy of Fig.~\ref{fig:T_to_M_to_H} for the ease of reading.}
	\label{fig:T_to_M_to_H_duplicated}
\end{figure}
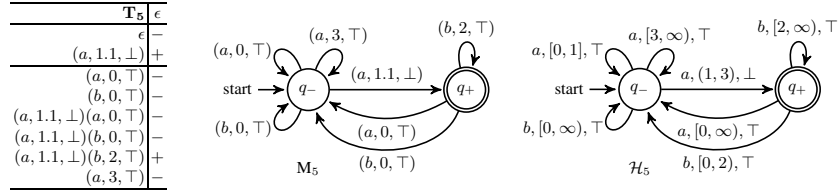

\paragraph*{Hypothesis construction for $\mathbf{T}_5$.}
Suppose $\mathbb{A}$ in Fig.~\ref{fig:dota_cota} recognizes the target timed language. Then the prepared table $\mathbf{T_5}$, the corresponding DFA $\text{M}_5$ and hypothesis $\mathcal{H}_5$ are depicted in Fig.~\ref{fig:T_to_M_to_H_duplicated} (a copy of Fig.~\ref{fig:T_to_M_to_H} for the ease of reading). The learner first builds the DFA $\text{M}_5$ as follows. For $\epsilon$ and $(a,1.1,\bot)$ in $\bm{S}$, set $Q_{M_5} = \{q_{-}, q_{+}\}$ with $\mathit{row}(\epsilon)=-$ and $\mathit{row}((a,1.1,\bot))=+$; while $q^0_{M_5}=q_{-}$; $F_{M_5}=\{q_{+}\}$ with $f((a,1.1,\bot)\cdot\epsilon)=+$; $\Sigma_{M_5}=\{(a,0,\top),(b,0,\top),(a,1.1,\bot),(b,2,\top),(a,3,\top)\}$ and in the mean time $\Delta_{M_5}=\{(q_{-},(a,1.1,\bot),q_{+}),(q_{-},(a,0,\top),q_{-}), (q_{-},(b,0,\top),q_{-}), (q_{+},(a,0,\top),q_{-}),\\ (q_{+},(b,0,\top),q_{-}), (q_{+},(b,2,\top),q_{-}), (q_{-},(a,3,\top),q_{-})\}$. The DFA $\text{M}_5$ is then transformed to $\mathcal{H}_5=(\Sigma,Q,q_0,F,c,\Delta)$ with the same set of locations, initial location and set of accepting locations. For location $q_{-}$, $\Psi_{q_{-},a} = \{0,3,1.1\}$. We arrange it to a list $\ell_{q_{-},a}=0,1.1,3$ and then $P^{c}(\ell_{q_{-},a})=\{[0,1],(1,3),[3,\infty)\}$. Then for transitions $(q_{-},(a,0,\top),q_{-})$,$(q_{-},(a,1.1,\bot),q_{+})$ and $(q_{-},(a,3,\top),q_{1})$ in $\Delta_{M_5}$, let $\delta_1=(q_{-},a,[0,1],\top,q_{-})$, $\delta_2=(q_{-},a,(1,3),\bot,q_{+})$ and $\delta_3=(q_{-},a,[3,\infty],\top,q_{-})$ be fresh transitions. 
The transformation for other locations and actions is analogous.
\paragraph*{The entire learning process.}
In Fig.~\ref{fig:ie_automata}, a prepared $\mathbf{T}_1$ is the initial instance of the table. The learner builds a DFA $\text{M}_1$ and a hypothesis $\mathcal{H}_1$ in Fig.~\ref{fig:ie_automata}. After making an equivalence query to the teacher, he receives a counterexample $\mathit{ctx}_1=((a,1.1,\bot),+)$. By transforming the delay-timed-word $(a,1.1,\bot)$ to a reset-logical-timed word $(a,1.1,\bot)$, the learner adds the normalized counterpart to the table and thus get the second instance $\mathbf{T}_2$ which is not closed. Hence, he moves $(a,1.1,\bot)$ from $\bm{R}$ to $\bm{S}$, and then adds reset-logical-timed words $(a,1.1,\bot)(a,0,\top)$ and $(a,1.1,\bot)(b,0,\top)$ to $\bm{R}$ after membership queries. The instance $\mathbf{T}_3$ is prepared. After two iterations, he arrives at $\mathbf{T}_5$. As described in Example~\ref{example:wrong_ctx}, the learner normalizes the transformed reset-logical-timed word $(a,0,\top)(a,1.3,\top)$ as $(a,0,\top)(a,1.1,\top)$. Then he gets $\mathbf{T}_6$ which is not consistent, since $\mathit{row}(\epsilon)=\mathit{row}((a,0,\top))$ and $\Pi_{\{1,2\}}((a,1.1,\bot))=\Pi_{\{1,2\}}((a,1.1,\top))$, but $\mathit{row}((a,1.1,\bot)) \neq \mathit{row}((a,0,\top)(a,1.1,\top))$. Hence, he adds $(a,1.1)\cdot\epsilon$ to $\bm{E}$ for $f((a,1.1,\bot)\cdot\epsilon) \neq f((a,0,\top)(a,1.1,\top) \cdot \epsilon)$ leading to $\mathbf{T}_7$. The process goes on until the learner finally gets a hypothesis $\mathcal{H}_{10}$ which recognizes the target timed language. Obviously, after combining transitions according to guards in $\mathcal{H}_{10}$,
we get a COTA same to $\mathbb{A}$ as depicted in Fig.~\ref{fig:dota_cota}.
\begin{figure}[H]
	\begin{center}
	\resizebox{\linewidth}{!}{
		\begin{tabular}{r|c}
			\hline
			$\mathbf{T_1}$ & $\epsilon$ \\ 
			\hline
			$\epsilon$ & $-$ \\
			\hline
			$(a,0,\top)$ & $-$ \\
			$(b,0,\top)$ & $-$ \\
			\hline
		\end{tabular}
		$\xLongrightarrow[g((a, 1.1, \bot))=(a, 1.1, \bot)]{\mathit{ctx}_1=(a, 1.1, \bot), +}$
		\begin{tabular}{r|c}
			\hline
			$\mathbf{T_2}$ & $\epsilon$ \\ 
			\hline
			$\epsilon$ & $-$ \\
			\hline
			$(a,0,\top)$ & $-$ \\
			$(b,0,\top)$ & $-$ \\
			$(a,1.1,\bot)$ & $+$ \\
			\hline
		\end{tabular}
		$\xLongrightarrow{\text{closed}}$
		\begin{tabular}{r|c}
			\hline
			$\mathbf{T_3}$ & $\epsilon$ \\ 
			\hline
			$\epsilon$ & $-$ \\
			$(a,1.1,\bot)$ & $+$ \\
			\hline
			$(a,0,\top)$ & $-$ \\
			$(b,0,\top)$ & $-$ \\
			$(a,1.1,\bot)(a,0,\top)$ & $-$ \\
			$(a,1.1,\bot)(b,0,\top)$ & $-$ \\
			\hline
		\end{tabular}
		$\xLongrightarrow[g((a,1.1,\bot)(b,2,\top))=(a,1.1,\bot)(b,2,\top)]{\mathit{ctx}_2=(a, 1.1, \bot)(b, 0.9, \top), +}$
		\begin{tabular}{r|c}
			\hline
			$\mathbf{T_4}$ & $\epsilon$ \\ 
			\hline
			$\epsilon$ & $-$ \\
			$(a,1.1,\bot)$ & $+$ \\
			\hline
			$(a,0,\top)$ & $-$ \\
			$(b,0,\top)$ & $-$ \\
			$(a,1.1,\bot)(a,0,\top)$ & $-$ \\
			$(a,1.1,\bot)(b,0,\top)$ & $-$ \\
			$(a,1.1,\bot)(b,2,\top)$ & $+$ \\
			\hline
		\end{tabular}
	}
	\resizebox{1\linewidth}{!}{
		$\xLongrightarrow[g((a,3,\top))=(a,3,\top)]{\mathit{ctx}_3=(a,3,\top),-}$
		\begin{tabular}{r|c}
			\hline
			$\mathbf{T_5}$ & $\epsilon$ \\ 
			\hline
			$\epsilon$ & $-$ \\
			$(a,1.1,\bot)$ & $+$ \\
			\hline
			$(a,0,\top)$ & $-$ \\
			$(b,0,\top)$ & $-$ \\
			$(a,1.1,\bot)(a,0,\top)$ & $-$ \\
			$(a,1.1,\bot)(b,0,\top)$ & $-$ \\
			$(a,1.1,\bot)(b,2,\top)$ & $+$ \\
			$(a,3,\top)$ & $-$ \\
			\hline
		\end{tabular}
		$\xLongrightarrow[g((a,0,\top)(a,1.3,\top))=(a,0,\top)(a,1.1,\top)]{\mathit{ctx}_4=(a,0,\top)(a,1.3,\top),-}$
		\begin{tabular}{r|c}
			\hline
			$\mathbf{T_6}$ & $\epsilon$ \\ 
			\hline
			$\epsilon$ & $-$ \\
			$(a,1.1,\bot)$ & $+$ \\
			\hline
			$(a,0,\top)$ & $-$ \\
			$(b,0,\top)$ & $-$ \\
			$(a,1.1,\bot)(a,0,\top)$ & $-$ \\
			$(a,1.1,\bot)(b,0,\top)$ & $-$ \\
			$(a,1.1,\bot)(b,2,\top)$ & $+$ \\
			$(a,3,\top)$ & $-$ \\
			$(a,0,\top)(a,1.1,\top)$ & $-$ \\
			\hline
		\end{tabular}
		$\xLongrightarrow{\text{consistent}}$
		\begin{tabular}{r|cc}
			\hline
			$\mathbf{T_7}$ & $\epsilon$ & $(a, 1.1)$ \\ 
			\hline
			$\epsilon$ & $-$ & $+$ \\
			$(a,1.1,\bot)$ & $+$ & $-$ \\
			\hline
			$(a,0,\top)$ & $-$ & $-$ \\
			$(b,0,\top)$ & $-$ & $+$ \\
			$(a,1.1,\bot)(a,0,\top)$ & $-$ & $-$ \\
			$(a,1.1,\bot)(b,0,\top)$ & $-$ & $-$ \\
			$(a,1.1,\bot)(b,2,\top)$ & $+$ & $-$ \\
			$(a,3,\top)$ & $-$ & $-$ \\
			$(a,0,\top)(a,1.1,\top)$ & $-$ & $-$ \\
			\hline
		\end{tabular}
	}
	\resizebox{\linewidth}{!}{
		$\xLongrightarrow{\text{evidence closed}}$
		\begin{tabular}{r|cc}
			\hline
			$\mathbf{T_8}$ & $\epsilon$ & $(a, 1.1)$ \\ 
			\hline
			$\epsilon$ & $-$ & $+$ \\
			$(a,1.1,\bot)$ & $+$ & $-$ \\
			\hline
			$(a,0,\top)$ & $-$ & $-$ \\
			$(b,0,\top)$ & $-$ & $+$ \\
			$(a,1.1,\bot)(a,0,\top)$ & $-$ & $-$ \\
			$(a,1.1,\bot)(b,0,\top)$ & $-$ & $-$ \\
			$(a,1.1,\bot)(b,2,\top)$ & $+$ & $-$ \\
			$(a,3,\top)$ & $-$ & $-$ \\
			$(a,0,\top)(a,1.1,\top)$ & $-$ & $-$ \\
			$(a,1.1,\bot)(a,1.1,\top)$ & $-$ & $-$ \\
			\hline
		\end{tabular}
		$\xLongrightarrow{\text{closed}}$
		\begin{tabular}{r|cc}
			\hline
			$\mathbf{T_9}$ & $\epsilon$ & $(a, 1.1)$ \\ 
			\hline
			$\epsilon$ & $-$ & $+$ \\
			$(a,1.1,\bot)$ & $+$ & $-$ \\
			$(a,0,\top)$ & $-$ & $-$ \\
			\hline
			$(b,0,\top)$ & $-$ & $+$ \\
			$(a,1.1,\bot)(a,0,\top)$ & $-$ & $-$ \\
			$(a,1.1,\bot)(b,0,\top)$ & $-$ & $-$ \\
			$(a,1.1,\bot)(b,2,\top)$ & $+$ & $-$ \\
			$(a,3,\top)$ & $-$ & $-$ \\
			$(a,0,\top)(a,1.1,\top)$ & $-$ & $-$ \\
			$(a,1.1,\bot)(a,1.1,\top)$ & $-$ & $-$ \\
			$(a,0,\top)(a,0,\top)$ & $-$ & $-$ \\
			$(a,0,\top)(b,0,\top)$ & $-$ & $-$ \\
			\hline
		\end{tabular}
		$\xLongrightarrow[g((a,1.1,\bot)(b,4,\top))=(a,1.1,\bot)(b,4,\top)]{\mathit{ctx}_5=(a,1.1,\bot)(b,2.9,\top),-}$
		\begin{tabular}{r|cc}
			\hline
			$\mathbf{T_{10}}$ & $\epsilon$ & $(a, 1.1)$ \\ 
			\hline
			$\epsilon$ & $-$ & $+$ \\
			$(a,1.1,\bot)$ & $+$ & $-$ \\
			$(a,0,\top)$ & $-$ & $-$ \\
			\hline
			$(b,0,\top)$ & $-$ & $+$ \\
			$(a,1.1,\bot)(a,0,\top)$ & $-$ & $-$ \\
			$(a,1.1,\bot)(b,0,\top)$ & $-$ & $-$ \\
			$(a,1.1,\bot)(b,2,\top)$ & $+$ & $-$ \\
			$(a,3,\top)$ & $-$ & $-$ \\
			$(a,0,\top)(a,1.1,\top)$ & $-$ & $-$ \\
			$(a,1.1,\bot)(a,1.1,\top)$ & $-$ & $-$ \\
			$(a,0,\top)(a,0,\top)$ & $-$ & $-$ \\
			$(a,0,\top)(b,0,\top)$ & $-$ & $-$ \\
			$(a,1.1,\bot)(b,4,\top)$ & $-$ & $-$ \\
			\hline
		\end{tabular}
	}\\[.3cm]
	\resizebox{1\linewidth}{!}{
	\begin{tabular}{ccccc}
		\hspace*{1cm}
		\begin{tikzpicture}[->, >=stealth', shorten >=1pt, auto, node distance=3cm, semithick, scale=0.9, every node/.style={scale=0.7}]
		\node[initial,state] (0) {$q_{-}$};
		\path (0) edge [in= 145, out=110, loop] node[left] {$(a,0,\top)$} (0)
		(0) edge [in= 70, out=35, loop] node[above] {$(b,0,\top)$} (0);
		\node [below=15pt, align=flush center,text width=3cm] at (0) {$\text{M}_1$};
		\end{tikzpicture}
		& &
		\begin{tikzpicture}[->, >=stealth', shorten >=1pt, auto, node distance=3cm, semithick, scale=0.9, every node/.style={scale=0.7}]
		\node[initial,state] (0) {$q_-$};
		\path (0) edge [in= 145, out=110, loop] node[left] {$a,[0,\infty),\top$} (0)
		(0) edge [in= 70, out=35, loop] node[above] {$b,[0,\infty),\top$} (0);
		\node [below=15pt, align=flush center,text width=3cm] at (0) {$\mathcal{H}_1$};
		\end{tikzpicture}
		& &
		\begin{tikzpicture}[->, >=stealth', shorten >=1pt, auto, node distance=3cm, semithick, scale=0.9, every node/.style={scale=0.7}]
		\node[initial,state] (0) {$q_{-}$};
		\node[accepting, state] (1) [right of = 0] {$q_{+}$};
		\path (0) edge [in= 145, out=110, loop] node[left] {$(a,0,\top)$} (0)
		(0) edge [in= -145, out=-110, loop] node[left] {$(b,0,\top)$} (0)
		(0) edge node[above] {$(a,1.1,\bot)$} (1)
		(1) edge[in=-30, out=-150] node[below] {$(a,0,\top)$} (0)
		(1) edge[in=-70, out=-110] node[below] {$(b,0,\top)$} (0);
		\node [below=25pt, align=flush center,text width=3cm] at (0) {$\text{M}_3$};
		\end{tikzpicture}
		\\
		\begin{tikzpicture}[->, >=stealth', shorten >=1pt, auto, node distance=3cm, semithick, scale=0.9, every node/.style={scale=0.7}]
		\node[initial,state] (0) {$q_{-}$};
		\node[accepting, state] (1) [right of = 0] {$q_{+}$};
		\path (0) edge [in= 145, out=110, loop] node[left] {$a,[0,1],\top$} (0)
		(0) edge [in= -145, out=-110, loop] node[left] {$b,[0,\infty),\top$} (0)
		(0) edge node[above] {$a,(1,\infty),\bot$} (1)
		(1) edge[in=-30, out=-150] node[below] {$a,[0,\infty),\top$} (0)
		(1) edge[in=-70, out=-110] node[below] {$b,[0,\infty),\top$} (0);
		\node [below=25pt, align=flush center,text width=3cm] at (0) {$\mathcal{H}_3$};
		\end{tikzpicture}
		& &
		\begin{tikzpicture}[->, >=stealth', shorten >=1pt, auto, node distance=3cm, semithick, scale=0.9, every node/.style={scale=0.7}]
		\node[initial,state] (0) {$q_{-}$};
		\node[accepting, state] (1) [right of = 0] {$q_{+}$};
		\path (0) edge [in= 145, out=110, loop] node[left] {$(a,0,\top)$} (0)
		(0) edge [in= -145, out=-110, loop] node[left] {$(b,0,\top)$} (0)
		(0) edge node[above] {$(a,1.1,\bot)$} (1)
		(1) edge[in=-30, out=-150] node[below] {$(a,0,\top)$} (0)
		(1) edge[in=-70, out=-110] node[below] {$(b,0,\top)$} (0)
		(1) edge[loop above] node[above] {$(b,2,\top)$} (1);
		\node [below=25pt, align=flush center,text width=3cm] at (0) {$\text{M}_4$};
		\end{tikzpicture}
		& &
		\begin{tikzpicture}[->, >=stealth', shorten >=1pt, auto, node distance=3cm, semithick, scale=0.9, every node/.style={scale=0.7}]
		\node[initial,state] (0) {$q_{-}$};
		\node[accepting, state] (1) [right of = 0] {$q_{+}$};
		\path (0) edge [in= 145, out=110, loop] node[left] {$a,[0,1],\top$} (0)
		(0) edge [in= -145, out=-110, loop] node[left] {$b,[0,\infty),\top$} (0)
		(0) edge node[above] {$a,(1,\infty),\bot$} (1)
		(1) edge[in=-30, out=-150] node[below] {$a,[0,\infty),\top$} (0)
		(1) edge[in=-70, out=-110] node[below] {$b,[0,2),\top$} (0)
		(1) edge[loop above] node[above] {$b,[2,\infty),\top$} (1);
		\node [below=25pt, align=flush center,text width=3cm] at (0) {$\mathcal{H}_4$};
		\end{tikzpicture}
		\\
		\begin{tikzpicture}[->, >=stealth', shorten >=1pt, auto, node distance=3cm, semithick, scale=0.9, every node/.style={scale=0.7}]
		\node[initial,state] (0) {$q_{-}$};
		\node[accepting, state] (1) [right of = 0] {$q_{+}$};
		\path (0) edge [in= 145, out=110, loop] node[left] {$(a,0,\top)$} (0)
		(0) edge [in= 35, out=70, loop] node[above] {$(a,3,\top)$} (0)
		(0) edge [in= -145, out=-110, loop] node[left] {$(b,0,\top)$} (0)
		(0) edge node[above] {$(a,1.1,\bot)$} (1)
		(1) edge[loop above] node[above] {$(b,2,\top)$} (1)
		(1) edge[in=-30, out=-150] node[below] {$(a,0,\top)$} (0)
		(1) edge[in=-70, out=-110] node[below] {$(b,0,\top)$} (0);
		\node [below=25pt, align=flush center,text width=3cm] at (0) {$\text{M}_5$};
		\end{tikzpicture}
		& &
		\begin{tikzpicture}[->, >=stealth', shorten >=1pt, auto, node distance=3cm, semithick, scale=0.9, every node/.style={scale=0.7}]
		\node[initial,state] (0) {$q_{-}$};
		\node[accepting, state] (1) [right of = 0] {$q_{+}$};
		\path (0) edge [in= 145, out=110, loop] node[left] {$a,[0,1],\top$} (0)
		(0) edge [in= 35, out=70, loop] node[above] {$a,[3,\infty),\top$} (0)
		(0) edge [in= -145, out=-110, loop] node[left] {$b,[0,\infty),\top$} (0)
		(0) edge node[above] {$a,(1,3),\bot$} (1)
		(1) edge[loop above] node[above] {$b,[2,\infty),\top$} (1)
		(1) edge[in=-30, out=-150] node[below] {$a,[0,\infty),\top$} (0)
		(1) edge[in=-70, out=-110] node[below] {$b,[0,2),\top$} (0);
		\node [below=25pt, align=flush center,text width=3cm] at (0) {$\mathcal{H}_5$};
		\end{tikzpicture}
		& &
		\begin{tikzpicture}[->, >=stealth', shorten >=1pt, auto, node distance=3.8cm, semithick, scale=0.9, every node/.style={scale=0.7}]
		\node[initial,state] (0) {$q_{-+}$};
		\node[accepting, state] (1) [right of = 0] {$q_{+-}$};
		\node[state] (2) [below right= 1.4cm and .9cm of 0] {$q_{--}$};
		\path  (0) edge node[above] {$(a,1.1,\bot)$} (1)
		(0) edge [loop above] node[above] {$(b,0,\top)$} (0)
		(1) edge [loop above] node[above] {$(b,2,\top)$} (1)
		(0) edge node[below,sloped]  {$(a,0,\top)$} (2)
		(0) edge [bend right] node[below,sloped,pos=.35]  {$(a,3,\top)$} (2)
		(1) edge [bend right] node[above,sloped,pos=.48]  {$(a,0,\top)$} (2)
		(1) edge node[above,sloped]  {$(a,1.1,\top)$} (2)
		(1) edge [bend left] node[below,sloped,pos=.45]  {$(b,0,\top)$} (2)
		(2) edge [in= -145, out=-110, loop] node[left] {$(a,0,\top)$} (2)
		(2) edge [in=175, out=-150,loop] node[left] {$(a,1.1,\top)$} (2)
		(2) edge [in= -70, out=-35, loop] node[right] {$(b,0,\top)$} (2);
		\node [below = 7pt  of 2, align=flush center,text width=3cm] {$\text{M}_9$};
		\end{tikzpicture}
		\\
		\begin{tikzpicture}[->, >=stealth', shorten >=1pt, auto, node distance=3.8cm, semithick, scale=0.9, every node/.style={scale=0.7}]
		\node[initial,state] (0) {$q_{-+}$};
		\node[accepting, state] (1) [right of = 0] {$q_{+-}$};
		\node[state] (2) [below right= 1.4cm and .9cm of 0] {$q_{--}$};
		\path  (0) edge node[above] {$a$, $(1,3)$, $\bot$} (1)
		(0) edge [loop above] node[above] {$b$, $[0,\infty)$, $\top$} (0)
		(1) edge [loop above] node[above] {$b$, $[2, \infty)$, $\top$} (1)
		(0) edge node[below,sloped]  {$a$, $[0,1]$, $\top$} (2)
		(0) edge [bend right] node[below,sloped,pos=.35]  {$a$, $[3,\infty)$, $\top$} (2)
		(1) edge [bend right] node[above,sloped,pos=.48]  {$a$, $[0,1]$, $\top$} (2)
		(1) edge node[above,sloped]  {$a$, $(1,\infty)$, $\top$} (2)
		(1) edge [bend left] node[below,sloped,pos=.45]  {$b$, $[0,2)$, $\top$} (2)
		(2) edge [in= -145, out=-110, loop] node[left] {$a$, $[0,1]$, $\top$} (2)
		(2) edge [in=175, out=-150,loop] node[left] {$a$, $(1,\infty)$, $\top$} (2)
		(2) edge [in= -70, out=-35, loop] node[right] {$b$, $[0,\infty)$, $\top$} (2);
		\node [below = 7pt  of 2, align=flush center,text width=3cm] {$\mathcal{H}_9$};
		\end{tikzpicture}
		& &
		\begin{tikzpicture}[->, >=stealth', shorten >=1pt, auto, node distance=3.8cm, semithick, scale=0.9, every node/.style={scale=0.7}]
		\node[initial,state] (0) {$q_{-+}$};
		\node[accepting, state] (1) [right of = 0] {$q_{+-}$};
		\node[state] (2) [below right= 1.4cm and .9cm of 0] {$q_{--}$};
		\path  (0) edge node[above] {$(a,1.1,\bot)$} (1)
		(0) edge [loop above] node[above] {$(b,0,\top)$} (0)
		(1) edge [loop above] node[above] {$(b,2,\top)$} (1)
		(0) edge node[below,sloped]  {$(a,0,\top)$} (2)
		(0) edge [bend right] node[below,sloped,pos=.35]  {$(a,3,\top)$} (2)
		(1) edge [bend right] node[above,sloped,pos=.48]  {$(a,0,\top)$} (2)
		(1) edge node[above,sloped]  {$(a,1.1,\top)$} (2)
		(1) edge [bend left] node[below,sloped,pos=.45]  {$(b,0,\top)$} (2)
		(1) edge [in=-10, out=-60] node[below,sloped,pos=.45]  {$(b,4,\top)$} (2)
		(2) edge [in= -145, out=-110, loop] node[left] {$(a,0,\top)$} (2)
		(2) edge [in=175, out=-150,loop] node[left] {$(a,1.1,\top)$} (2)
		(2) edge [in= -70, out=-35, loop] node[right] {$(b,0,\top)$} (2);
		\node [below = 7pt  of 2, align=flush center,text width=3cm] {$\text{M}_{10}$};
		\end{tikzpicture}
		& &
		\begin{tikzpicture}[->, >=stealth', shorten >=1pt, auto, node distance=3.8cm, semithick, scale=0.9, every node/.style={scale=0.7}]
		\node[initial,state] (0) {$q_{-+}$};
		\node[accepting, state] (1) [right of = 0] {$q_{+-}$};
		\node[state] (2) [below right= 1.4cm and .9cm of 0] {$q_{--}$};
		\path  (0) edge node[above] {$a$, $(1,3)$, $\bot$} (1)
		(0) edge [loop above] node[above] {$b$, $[0,\infty)$, $\top$} (0)
		(1) edge [loop above] node[above] {$b$, $[2,4)$, $\top$} (1)
		(0) edge node[below,sloped]  {$a$, $[0,1]$, $\top$} (2)
		(0) edge [bend right] node[below,sloped,pos=.35]  {$a$, $[3,\infty)$, $\top$} (2)
		(1) edge [bend right] node[above,sloped,pos=.48]  {$a$, $[0,1]$, $\top$} (2)
		(1) edge node[above,sloped]  {$a$, $(1,\infty)$, $\top$} (2)
		(1) edge [bend left] node[below,sloped,pos=.45]  {$b$, $[0,2)$, $\top$} (2)
		(1) edge [in=-10, out=-60] node[below,sloped,pos=.45]  {$(b,[4,\infty),\top)$} (2)
		(2) edge [in= -145, out=-110, loop] node[left] {$a$, $[0,1]$, $\top$} (2)
		(2) edge [in=175, out=-150,loop] node[left] {$a$, $(1,\infty)$, $\top$} (2)
		(2) edge [in= -70, out=-35, loop] node[right] {$b$, $[0,\infty)$, $\top$} (2);
		\node [below = 7pt  of 2, align=flush center,text width=3cm] {$\mathcal{H}_{10}$};
		\end{tikzpicture}
	\end{tabular}
	}
	\end{center}
	\caption{Iterations of the timed observation table, DFAs and hypotheses w.r.t. $\mathcal{A}$ in Fig.~\ref{fig:dota_cota}.}
	\label{fig:ie_automata}
\end{figure}
\oomit{
\noindent the learner adds the normalized counterpart to the table and thus get the second instance $\mathbf{T}_2$ which is not closed. Hence, he moves $(a,1.1,\bot)$ from $\bm{R}$ to $\bm{S}$, and then adds reset-logical-timed words $(a,1.1,\bot)(a,0,\top)$ and $(a,1.1,\bot)(b,0,\top)$ to $\bm{R}$ after membership queries. The instance $\mathbf{T}_3$ is prepared. After two iterations, he arrives at $\mathbf{T}_5$. As described in Example~\ref{example:wrong_ctx}, the learner normalizes the transformed reset-logical-timed word $(a,0,\top)(a,1.3,\top)$ as $(a,0,\top)(a,1.1,\top)$. Then he gets $\mathbf{T}_6$ which is not consistent, since $\mathit{row}(\epsilon)=\mathit{row}((a,0,\top))$ and $\Pi_{\{1,2\}}((a,1.1,\bot))=\Pi_{\{1,2\}}((a,1.1,\top))$, but $\mathit{row}((a,1.1,\bot)) \neq \mathit{row}((a,0,\top)(a,1.1,\top))$. Hence, he adds $(a,1.1)\cdot\epsilon$ to $\bm{E}$ for $f((a,1.1,\bot)\cdot\epsilon) \neq f((a,0,\top)(a,1.1,\top) \cdot \epsilon)$ leading to $\mathbf{T}_7$. The process goes on until the learner finally gets a hypothesis $\mathcal{H}_{10}$ which recognizes the target timed language. Obviously, after combining transitions according to guards in $\mathcal{H}_{10}$,
we get a COTA same to $\mathbb{A}$ as depicted in Fig.~\ref{fig:dota_cota}.
}

\section{Algorithm for Equivalence Queries}\label{appendix_algs}

\begin{algorithm}[h]
	\caption{equivalence\_query($\mathcal{H}$)}
	\label{alg:equivalent}
	\SetKwInOut{Input}{input}
	\SetKwInOut{Output}{output}
	\Input{a hypothesis $\mathcal{H}$.}
	\Output{$\mathit{equivalent}$ : a Boolean value to identify whether $\mathcal{L}(\mathcal{H}) = \mathcal{L}(\mathbb{A})$ where COTA $\mathbb{A}$ recognizes the target language;\\
		$\mathit{ctx}$ : a counterexample.}
	$\mathit{equivalent}$ $\leftarrow$ $\bot$; $\mathit{ctx}$ $\leftarrow$ $\epsilon$\;
	$\mathit{flag}_{-},\mathit{flag}_{+} \leftarrow$ $\top$\;
	\If{$\mathcal{L}(\mathcal{H})\nsubseteq\mathcal{L}(\mathbb{A})$}{ 
		$\mathit{flag}_{-}$ $\leftarrow$ $\bot$ \tcp*{negative counterexample}
		generate a reset-delay-timed word $\omega_{r}$ from a bad configration $W$\;
		$\mathit{ctx}_{-}$ $\leftarrow$ $(\omega_{r}, -)$\;
	}
	\If{$\mathcal{L}(\mathbb{A})\nsubseteq\mathcal{L}(\mathcal{H})$}{
		$\mathit{flag}_{+}$ $\leftarrow$ $\bot$ \tcp*{positive counterexample}
		generate a reset-delay-timed word ${\omega_{r}}'$ from a bad configration $W'$\;
		$\mathit{ctx}_{+}$ $\leftarrow$ $({\omega_{r}}', +)$\;
	}
	$\mathit{equivalent}$ $\leftarrow$ $\mathit{flag}_{-} \wedge \mathit{flag}_{+}$\;
	\If{$\mathit{equivalent}$ = $\bot$}{
		$\mathit{ctx}$ $\leftarrow$ select a counterexample from $\mathit{ctx}_{+}$ and $\mathit{ctx}_{-}$\;
	}
	\Return $\mathit{equivalent}$, $\mathit{ctx}$\;
\end{algorithm}

\section{Automata Pertaining to the TCP Protocol}\label{appendix_tcp}

Relabelling actions in~\cite{rfc793} results in the automaton to be learned (left of Fig.~\ref{fig:tcp}):
\begin{align*}
\{&a\colon \mathtt{passive\ OPEN}, b\colon \mathtt{rcv\ SYN}, c\colon \mathtt{SEND}, d\colon \mathtt{rcv\ SYN,ACK}, e\colon \mathtt{rcv\ ACK\ of\ SYN},\\
&f\colon \mathtt{CLOSE}, g\colon \mathtt{rcv\ FIN}, h\colon \mathtt{rcv\ ACK\ of\ FIN}, i\colon \mathtt{Timeout}, j\colon \mathtt{active\ OPEN}\}.
\end{align*}

\noindent Mapping locations in the learnt automaton (right of Fig.~\ref{fig:tcp}) back to that in~\cite{rfc793}:
\begin{align*}
\{&q_1\colon \mathtt{CLOSED}, q_2\colon \mathtt{LISTEN}, q_3\colon \mathtt{SYN\ SENT}, q_4\colon \mathtt{SYN\ RCVD}, q_5\colon \mathtt{ESTAB},\\
&q_6, q_{15}\colon \mathtt{FINWAIT-1}, q_7, q_{14}\colon \mathtt{CLOSE\ WAIT}, q_8, q_{13}\colon \mathtt{CLOSING},\\ &q_9, q_{12}\colon \mathtt{FINWAIT-2}, q_{10}\colon \mathtt{LAST-ACK}, q_{11}\colon \mathtt{TIME\ WAIT}\}.
\end{align*}

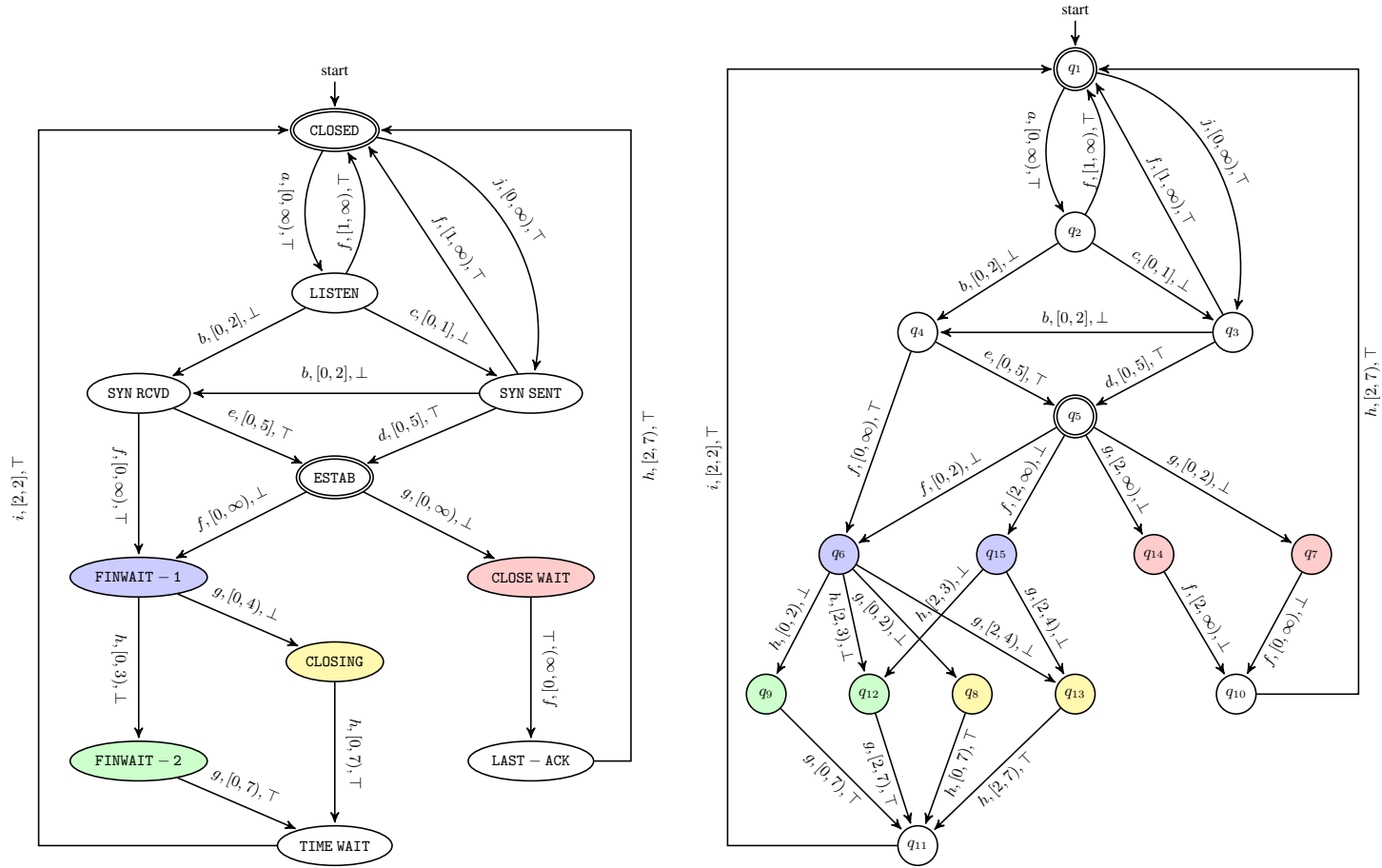
\begin{sidewaysfigure}
	\begin{minipage}[b]{.49\linewidth}
		\centering
		\begin{tikzpicture}[->, >=stealth', shorten >=1pt, auto, node distance=2.7cm, semithick, scale = 0.9,every node/.style={ellipse, scale=0.7}]
		\node[initial above, accepting, state, ellipse]  (closed) {$\mathtt{CLOSED}$};
		\node[state, ellipse](listen) [below = 1.7cm of closed] {$\mathtt{LISTEN}$};
		\node[state, ellipse](synrcvd) [below left = 1cm and 1.8cm of listen] {$\mathtt{SYN\ RCVD}$};
		\node[state, ellipse](synsent) [below right = 1cm and 1.8cm of listen] {$\mathtt{SYN\ SENT}$};
		\node[accepting, state, ellipse](estab) [below = 2cm of listen] {$\mathtt{ESTAB}$};
		\node[state, ellipse, fill=blue!20](finwait1) [below = 2cm of synrcvd] {$\mathtt{FINWAIT-1}$};
		\node[state, ellipse, fill=red!20](closewait) [below = 2cm of synsent] {$\mathtt{CLOSE\ WAIT}$};
		\node[state, ellipse, fill=yellow!40](closing) [below = 2cm of estab] {$\mathtt{CLOSING}$};
		\node[state, ellipse, fill=green!20](finwait2) [below = 2cm of finwait1] {$\mathtt{FINWAIT-2}$};
		\node[state, ellipse](lastack) [below = 2cm of closewait] {$\mathtt{LAST-ACK}$};
		\node[state, ellipse](timewait) [below = 2cm of closing] {$\mathtt{TIME\ WAIT}$};
		
		\path (closed) edge [out=-120, in=120] node[below, sloped, pos=.5] {$a, [0,\infty), \top$} (listen)
		(listen) edge [out=60, in=-60] node[above, sloped, pos=.5] {$f, [1,\infty), \top$} (closed)
		(listen) edge node[sloped, above, pos=0.5] {$b, [0,2], \bot$} (synrcvd)
		(listen) edge node[sloped, above, pos=0.5] {$c, [0,1], \bot$} (synsent)
		(synsent) edge node[sloped, above, pos=0.5] {$b, [0,2], \bot$} (synrcvd)
		(synsent) edge node[sloped, above, pos=0.6] {$d, [0,5], \top$} (estab)
		(closed) edge[out=-10, in=80] node[above, sloped, pos=.5] {$j, [0,\infty), \top$} (synsent)
		(synrcvd) edge node[sloped, above, pos=0.6] {$e, [0,5], \top$} (estab)
		(synrcvd) edge node[sloped, below, pos=0.5] {$f, [0,\infty), \top$} (finwait1)
		(estab) edge node[sloped, above, pos=0.5] {$f, [0,\infty), \bot$} (finwait1)
		(estab) edge node[sloped, above, pos=0.5] {$g, [0,\infty), \bot$} (closewait)
		(finwait1) edge node[sloped, below, pos=0.5] {$h, [0,3), \bot$} (finwait2)
		(finwait1) edge node[sloped, above, pos=0.5] {$g, [0,4), \bot$} (closing)
		(finwait2) edge node[sloped, above, pos=0.5] {$g, [0,7), \top$} (timewait)
		(closewait) edge node[sloped, below, pos=0.5] {$f, [0,\infty), \bot$} (lastack)
		(closing) edge node[sloped, above, pos=0.5] {$h, [0,7), \top$} (timewait)
		;
		\draw[semithick] (synsent) -- node[above, sloped, pos=.5] {$f, [1,\infty), \top$} (.5,-.22);
		\draw[semithick] (lastack.east) -| (4.6,0) node [midway, below, sloped, pos=.75] (eastedge) {$h, [2,7), \top$} -- (closed.east);
		\draw[semithick] (timewait.west) -| (-4.6,0) node [midway, above, sloped, pos=.75] (westedge) {$i, [2,2], \top$} -- (closed.west); 
		\end{tikzpicture}
	\end{minipage}
	\hfill
	\begin{minipage}[b]{.49\linewidth}
		\centering
		\begin{tikzpicture}[->, >=stealth', shorten >=1pt, auto, node distance=2.7cm, semithick, scale = 0.9,every node/.style={scale=0.7}]
		\node[initial above, accepting, state]  (q1) {$q_1$};
		\node[state](q2) [below = 1.7cm of q1] {$q_2$};
		\node[state](q4) [below left = 1cm and 1.8cm of q2] {$q_4$};
		\node[state](q3) [below right = 1cm and 1.8cm of q2] {$q_3$};
		\node[accepting, state](q5) [below = 2cm of q2] {$q_5$};
		\node[state, fill=blue!20](q6) [below left = 2.7cm and 0.7cm of q4] {$q_6$};
		\node[state, fill=blue!20](q15) [below right = 2.7cm and 0.7cm of q4] {$q_{15}$};
		\node[state, fill=red!20](q7) [below right = 2.7cm and 0.7cm of q3] {$q_7$};
		\node[state, fill=red!20](q14) [below left = 2.7cm and 0.7cm of q3] {$q_{14}$};
		\node[state, fill=green!20](q12) at (-3.2,-9.7) {$q_{12}$};
		\node[state, fill=green!20](q9) at (-4.8,-9.7) {$q_{9}$};
		\node[state, fill=yellow!40](q8) at (-1.6,-9.7) {$q_8$};
		\node[state, fill=yellow!40](q13) at (0,-9.7) {$q_{13}$};
		\node[state](q10) at (2.5,-9.7) {$q_{10}$};
		\node[state](q11) [below = 6.6cm of q4] {$q_{11}$};
		
		\path (q1) edge [out=-120, in=120] node[below, sloped, pos=.5] {$a, [0,\infty), \top$} (q2)
		(q2) edge [out=60, in=-60] node[above, sloped, pos=.5] {$f, [1,\infty), \top$} (q1)
		(q2) edge node[sloped, above, pos=0.5] {$b, [0,2], \bot$} (q4)
		(q2) edge node[sloped, above, pos=0.5] {$c, [0,1], \bot$} (q3)
		(q3) edge node[sloped, above, pos=0.5] {$b, [0,2], \bot$} (q4)
		(q3) edge node[sloped, above, pos=0.6] {$d, [0,5], \top$} (q5)
		(q4) edge node[sloped, above, pos=0.6] {$e, [0,5], \top$} (q5)
		(q1) edge[out=-10, in=80] node[above, sloped, pos=.5] {$j, [0,\infty), \top$} (q3)
		(q4) edge node[sloped, above, pos=0.5] {$f, [0,\infty), \top$} (q6)
		(q5) edge node[sloped, above, pos=0.5] {$f, [0,2), \bot$} (q6)
		(q5) edge node[sloped, above, pos=0.5] {$f, [2,\infty), \bot$} (q15)
		(q5) edge node[sloped, above, pos=0.5] {$g, [0,2), \bot$} (q7)
		(q5) edge node[sloped, above, pos=0.5] {$g, [2,\infty), \bot$} (q14)
		(q7) edge node[sloped, below, pos=0.5] {$f, [0,\infty), \bot$} (q10)
		(q14) edge node[sloped, above, pos=0.5] {$f, [2,\infty), \bot$} (q10)
		(q6) edge node[sloped, above, pos=0.5] {$h, [0,2), \bot$} (q9)
		(q6) edge node[sloped, below, pos=0.5] {$h, [2,3), \bot$} (q12)
		(q6) edge node[sloped, below, pos=0.36] {$g, [0,2), \bot$} (q8)
		(q6) edge node[sloped, above, pos=0.7] {$g, [2,4), \bot$} (q13)
		(q15) edge node[sloped, above, pos=0.3] {$h, [2,3), \bot$} (q12)
		(q15) edge node[sloped, above, pos=0.5] {$g, [2,4), \bot$} (q13)
		(q9) edge node[sloped, below, pos=0.5] {$g, [0,7), \top$} (q11)
		(q12) edge node[sloped, below, pos=0.5] {$g, [2,7), \top$} (q11)
		(q8) edge node[sloped, below, pos=0.4] {$h, [0,7), \top$} (q11)
		(q13) edge node[sloped, below, pos=0.5] {$h, [2,7), \top$} (q11);
		
		\draw[semithick] (q3) -- node[above, sloped, pos=.5] {$f, [1,\infty), \top$} (.3,-.18);
		\draw[semithick] (q10.east) -| (4.4,0) node [midway, below, sloped, pos=.75] (eastedge1) {$h, [2,7), \top$} -- (q1.east);
		\draw[semithick] (q11.west) -| (-5.4,0) node [midway, above, sloped, pos=.75] (westedge1) {$i, [2,2], \top$} -- (q1.west); 
		
		\end{tikzpicture}
	\end{minipage}
	\caption{Left: The functional specification of the TCP protocol with timing constraints. Right: The learnt functional specification of the TCP protocol. Colors indicate the splitting of locations incurred in the learning process.}
	\label{fig:tcp}
\end{sidewaysfigure}

\oomit{
	\begin{figure}[h]
		\centering
		\begin{tikzpicture}[->, >=stealth', shorten >=1pt, auto, node distance=2.7cm, semithick, scale = 0.9,every node/.style={ellipse, scale=0.7}]
		\node[initial above, accepting, state, ellipse]  (closed) {$\mathtt{CLOSED}$};
		\node[state, ellipse](listen) [below = 1.7cm of closed] {$\mathtt{LISTEN}$};
		\node[state, ellipse](synrcvd) [below left = 1cm and 1.8cm of listen] {$\mathtt{SYN\ RCVD}$};
		\node[state, ellipse](synsent) [below right = 1cm and 1.8cm of listen] {$\mathtt{SYN\ SENT}$};
		\node[accepting, state, ellipse](estab) [below = 2cm of listen] {$\mathtt{ESTAB}$};
		\node[state, ellipse, fill=blue!20](finwait1) [below = 2cm of synrcvd] {$\mathtt{FINWAIT-1}$};
		\node[state, ellipse, fill=red!20](closewait) [below = 2cm of synsent] {$\mathtt{CLOSE\ WAIT}$};
		\node[state, ellipse, fill=yellow!40](closing) [below = 2cm of estab] {$\mathtt{CLOSING}$};
		\node[state, ellipse, fill=green!20](finwait2) [below = 2cm of finwait1] {$\mathtt{FINWAIT-2}$};
		\node[state, ellipse](lastack) [below = 2cm of closewait] {$\mathtt{LAST-ACK}$};
		\node[state, ellipse](timewait) [below = 2cm of closing] {$\mathtt{TIME\ WAIT}$};
		
		\path (closed) edge [out=-120, in=120] node[below, sloped, pos=.5] {$a, [0,\infty), \top$} (listen)
		(listen) edge [out=60, in=-60] node[above, sloped, pos=.5] {$f, [1,\infty), \top$} (closed)
		(listen) edge node[sloped, above, pos=0.5] {$b, [0,2], \bot$} (synrcvd)
		(listen) edge node[sloped, above, pos=0.5] {$c, [0,1], \bot$} (synsent)
		(synsent) edge node[sloped, above, pos=0.5] {$b, [0,2], \bot$} (synrcvd)
		(synsent) edge node[sloped, above, pos=0.6] {$d, [0,5], \top$} (estab)
		(closed) edge[out=-10, in=80] node[above, sloped, pos=.5] {$j, [0,\infty), \top$} (synsent)
		(synrcvd) edge node[sloped, above, pos=0.6] {$e, [0,5], \top$} (estab)
		(synrcvd) edge node[sloped, below, pos=0.5] {$f, [0,\infty), \top$} (finwait1)
		(estab) edge node[sloped, above, pos=0.5] {$f, [0,\infty), \bot$} (finwait1)
		(estab) edge node[sloped, above, pos=0.5] {$g, [0,\infty), \bot$} (closewait)
		(finwait1) edge node[sloped, below, pos=0.5] {$h, [0,3), \bot$} (finwait2)
		(finwait1) edge node[sloped, above, pos=0.5] {$g, [0,4), \bot$} (closing)
		(finwait2) edge node[sloped, above, pos=0.5] {$g, [0,7), \top$} (timewait)
		(closewait) edge node[sloped, below, pos=0.5] {$f, [0,\infty), \bot$} (lastack)
		(closing) edge node[sloped, above, pos=0.5] {$h, [0,7), \top$} (timewait)
		;
		\draw[semithick] (synsent) -- node[above, sloped, pos=.5] {$f, [1,\infty), \top$} (.5,-.22);
		\draw[semithick] (lastack.east) -| (4.6,0) node [midway, below, sloped, pos=.75] (eastedge) {$h, [2,7), \top$} -- (closed.east);
		\draw[semithick] (timewait.west) -| (-4.6,0) node [midway, above, sloped, pos=.75] (westedge) {$i, [2,2], \top$} -- (closed.west); 
		\end{tikzpicture}
		\caption{The functional specification of the TCP protocol with timing constraints.}
		\label{fig:tcp_A}
	\end{figure}
}

\oomit{
\begin{figure}[h]
	\centering
	\begin{tikzpicture}[->, >=stealth', shorten >=1pt, auto, node distance=2.7cm, semithick, scale = 0.9,every node/.style={scale=0.7}]
	\node[initial above, accepting, state]  (q1) {$q_1$};
	\node[state](q2) [below = 1.7cm of q1] {$q_2$};
	\node[state](q4) [below left = 1cm and 1.8cm of q2] {$q_4$};
	\node[state](q3) [below right = 1cm and 1.8cm of q2] {$q_3$};
	\node[accepting, state](q5) [below = 2cm of q2] {$q_5$};
	\node[state, fill=blue!20](q6) [below left = 2.7cm and 0.7cm of q4] {$q_6$};
	\node[state, fill=blue!20](q15) [below right = 2.7cm and 0.7cm of q4] {$q_{15}$};
	\node[state, fill=red!20](q7) [below right = 2.7cm and 0.7cm of q3] {$q_7$};
	\node[state, fill=red!20](q14) [below left = 2.7cm and 0.7cm of q3] {$q_{14}$};
	\node[state, fill=green!20](q12) at (-3.2,-9.7) {$q_{12}$};
	\node[state, fill=green!20](q9) at (-4.8,-9.7) {$q_{9}$};
	\node[state, fill=yellow!40](q8) at (-1.6,-9.7) {$q_8$};
	\node[state, fill=yellow!40](q13) at (0,-9.7) {$q_{13}$};
	\node[state](q10) at (2.5,-9.7) {$q_{10}$};
	\node[state](q11) [below = 6.6cm of q4] {$q_{11}$};

	\path (q1) edge [out=-120, in=120] node[below, sloped, pos=.5] {$a, [0,\infty), \top$} (q2)
	(q2) edge [out=60, in=-60] node[above, sloped, pos=.5] {$f, [1,\infty), \top$} (q1)
	(q2) edge node[sloped, above, pos=0.5] {$b, [0,2], \bot$} (q4)
	(q2) edge node[sloped, above, pos=0.5] {$c, [0,1], \bot$} (q3)
	(q3) edge node[sloped, above, pos=0.5] {$b, [0,2], \bot$} (q4)
	(q3) edge node[sloped, above, pos=0.6] {$d, [0,5], \top$} (q5)
	(q4) edge node[sloped, above, pos=0.6] {$e, [0,5], \top$} (q5)
	(q1) edge[out=-10, in=80] node[above, sloped, pos=.5] {$j, [0,\infty), \top$} (q3)
	(q4) edge node[sloped, above, pos=0.5] {$f, [0,\infty), \top$} (q6)
	(q5) edge node[sloped, above, pos=0.5] {$f, [0,2), \bot$} (q6)
	(q5) edge node[sloped, above, pos=0.5] {$f, [2,\infty), \bot$} (q15)
	(q5) edge node[sloped, above, pos=0.5] {$g, [0,2), \bot$} (q7)
	(q5) edge node[sloped, above, pos=0.5] {$g, [2,\infty), \bot$} (q14)
	(q7) edge node[sloped, below, pos=0.5] {$f, [0,\infty), \bot$} (q10)
	(q14) edge node[sloped, above, pos=0.5] {$f, [2,\infty), \bot$} (q10)
	(q6) edge node[sloped, above, pos=0.5] {$h, [0,2), \bot$} (q9)
	(q6) edge node[sloped, below, pos=0.5] {$h, [2,3), \bot$} (q12)
	(q6) edge node[sloped, below, pos=0.36] {$g, [0,2), \bot$} (q8)
	(q6) edge node[sloped, above, pos=0.7] {$g, [2,4), \bot$} (q13)
	(q15) edge node[sloped, above, pos=0.3] {$h, [2,3), \bot$} (q12)
	(q15) edge node[sloped, above, pos=0.5] {$g, [2,4), \bot$} (q13)
	(q9) edge node[sloped, below, pos=0.5] {$g, [0,7), \top$} (q11)
	(q12) edge node[sloped, below, pos=0.5] {$g, [2,7), \top$} (q11)
	(q8) edge node[sloped, below, pos=0.4] {$h, [0,7), \top$} (q11)
	(q13) edge node[sloped, below, pos=0.5] {$h, [2,7), \top$} (q11);
	
	\draw[semithick] (q3) -- node[above, sloped, pos=.5] {$f, [1,\infty), \top$} (.3,-.18);
	\draw[semithick] (q10.east) -| (4.4,0) node [midway, below, sloped, pos=.75] (eastedge1) {$h, [2,7), \top$} -- (q1.east);
	\draw[semithick] (q11.west) -| (-5.4,0) node [midway, above, sloped, pos=.75] (westedge1) {$i, [2,2], \top$} -- (q1.west); 

	\end{tikzpicture}
	\caption{The learnt functional specification of the TCP protocol.}
	\label{fig:tcp_H}
\end{figure}
}

\end{subappendices}
\oomit{

\vfill

{\small\medskip\noindent{\bf Open Access} This chapter is licensed under the terms of the Creative Commons Attribution 4.0 International License (\url{http://creativecommons.org/licenses/by/4.0/}), which permits use, sharing, adaptation, distribution and reproduction in any medium or format, as long as you give appropriate credit to the original author(s) and the source, provide a link to the Creative Commons license and indicate if changes were made.}

{\small \spaceskip .28em plus .1em minus .1em The images or other third party material in this chapter are included in the chapter's Creative Commons license, unless indicated otherwise in a credit line to the material.~If material is not included in the chapter's Creative Commons license and your intended use is not permitted by statutory regulation or exceeds the permitted use, you will need to obtain permission directly from the copyright holder.}

\medskip\noindent\includegraphics{cc_by_4-0.eps}
}
\end{document}